\documentclass[final,onecolumn,12pt]{IEEEtran}
\usepackage{graphicx,psfrag,epsfig,epsf,latexsym,hhline,amsmath,amssymb,multirow}
\usepackage[usenames,dvipsnames]{pstricks}
\usepackage{pst-plot}
\usepackage{pstricks-add}
\usepackage{graphicx}
\usepackage{amsthm}
\usepackage{algorithm}
\usepackage{algpseudocode}

\usepackage{blindtext}
\usepackage{etoolbox}

\graphicspath{ {figures/} }

\input{Jerry.def}

\begin{document}
\title{Lattice Index Codes from Algebraic Number Fields}
\author{Yu-Chih Huang\\
Department of Communication Engineering \\
National Taipei University\\
{\tt\small {\{ychuang@mail.ntpu.edu.tw\}} }
\thanks{This work was published in part at the 2015 IEEE International Symposium on Information Theory (ISIT), Hong Kong. This work was supported by the Ministry of Science and Technology, Taiwan, under Grant MOST 104-2218-E-305-001-MY2.}}

\maketitle

\begin{abstract}
    Broadcasting $K$ independent messages to multiple users where each user demands all the messages and has a subset of the messages as side information is studied. Recently, Natarajan, Hong, and Viterbo proposed a novel broadcasting strategy called lattice index coding which uses lattices constructed over some principal ideal domains (PIDs) for transmission and showed that this scheme provides uniform side information gains. In this paper, we generalize this strategy to general rings of algebraic integers of number fields which may not be PIDs. Upper and lower bounds on the side information gains for the proposed scheme constructed over some interesting classes of number fields are provided and are shown to coincide asymptotically in message rates. This generalization substantially enlarges the design space and partially includes the scheme by Natarajan, Hong, and Viterbo as a special case. Perhaps more importantly, in addition to side information gains, the proposed lattice index codes benefit from diversity gains inherent in constellations carved from number fields when used over Rayleigh fading channel. Some interesting examples are also provided for which the proposed scheme allows all the messages to be from the same field.
\end{abstract}

\section{Introduction}
Broadcast with receiver message side information has recently attracted a lot of attention in the network layer with many interesting and important advances. For example, in the index coding problem \cite{birk98} \cite{bar-yossef11} \cite{rouayheb10}, a sender broadcasts a set of independent messages to several receivers where each receiver demands a subset of messages and has another subset of messages as side information. Another excellent example is the caching problem \cite{niesen14} where a sender broadcasts a set of independent messages with a fraction of messages (or functions of messages) being prefetched into receivers beforehand during off-peak hours. For such problems, it has been shown that a carefully designed broadcasting strategy which allows the receivers to better exploit side information can substantially improve the system throughput.

As the physical layer counterpart, Gaussian broadcast channels with receiver side information have recently been popular as well. In \cite{wu07}, the two-user Gaussian broadcast channel with receiver message side information was studied and the capacity region was fully characterized for all (5 in total) possible side information configurations. Yoo \textit{et al.} in \cite{tie09} considered the three-user scenario where they showed that a separation-based scheme which separately employs index coding and physical layer coding can achieve the capacity region to within a constant gap for all side information configurations regardless of channel parameters. For the three-user case, the capacity region was characterized for some particular side information configurations \cite{sima14} \cite{asadi14}. For the case having more than three users, our knowledge is fairly limited. This problem has also been independently studied in the context of the broadcast phase of the two-way (or multi-way) relay channel \cite{oechtering08}.

Recently, in \cite{viterbo14index}, Natarajan \textit{et al.} considered a special class of Gaussian broadcast channels with receiver message side information as shown in Fig.~\ref{fig:system_model}. In this model, each receiver demands \textit{all the messages} and can have an arbitrary subset of messages as side information. Unlike other work focusing on the capacity region, \cite{viterbo14index} focused on designing practical codes/modulations and proposed the so-called lattice index codes that provide large side information gains for any side information configuration. In \cite{viterbo_isit15}, they further proposed the index coded modulation which adopts powerful linear codes as outer codes in conjunction with index modulation as inner code to enjoy coding gains on top of side information gains.
\begin{figure}
    \centering
    \includegraphics[width=3.5in]{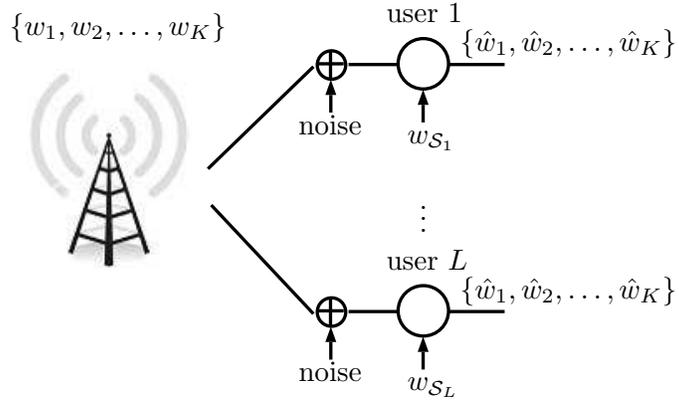}
    \caption{Index coding over AWGN channel where there are $K$ messages and $L$ users. The user $l$ has a subset of message $w_{\mc{S}_l}$ as side information.}
    \label{fig:system_model}
\end{figure}

In this paper, we study the same special class of Gaussian broadcast channels with receiver message side information as that in \cite{viterbo14index} and also focus on practical code/modulation design. The contributions of this paper is summarized in the following:
\begin{itemize}
    \item In Section~\ref{sec:proposed_const}, we generalize the lattice index coding scheme to a general ring of algebraic integers. This substantially expands the design space as the scheme in \cite{viterbo14index} only considers $\mbb{Z}$, $\Zi$, $\Zw$, and $\mbb{H}$. This is a nontrivial generalization as the rings mentioned above are all principal ideal domains (PID) while a general ring of algebraic integers may not be so. Our generalization partially subsumes the lattice index codes in \cite{viterbo14index} as special cases (the design over $\mbb{H}$ is not covered).
    \item In Section~\ref{sec:proposed_property}, we provide bounds on the side information gains for the proposed scheme designed over some interesting families of number fields. This includes totally real number fields in Theorem~\ref{thm:si_gain_real}, totally complex number fields in Theorem~\ref{thm:si_gain_complex}, and imaginary quadratic integers happening to be PID in Theorem~\ref{thm:si_gain_Q}. The bounds are shown to coincide at least asymptotically in message rates.
    \item We also consider the same broadcast problem over Rayleigh fading channel where we show in Section~\ref{sec:proposed_property} that the proposed lattice codes exhibit diversity gains in addition to side information gains for any given message side information. In addition, for the proposed scheme constructed from a totally real number field, we further show that higher side information gains can be obtained as the minimum product distance is also increased when revealing messages to the receiver.
    \item It has been pointed out in \cite{viterbo14index} that one drawback of the scheme therein is that messages are from different fields; however, in many applications, one may want their messages to be from the same field. In Section~\ref{sec:design}, we provide some interesting design examples to demonstrate that as a consequence of having a larger design space, one can easily find instances of the proposed schemes where messages are from the same field.
    \item In Section~\ref{sec:CIM}, similar to \cite{viterbo_isit15}, we realize coded index modulation with the proposed lattice index modulation as inner code and low-density parity-check (LDPC) codes as outer codes to obtain coding gains on top of side information gains.
\end{itemize}

\subsection{Organization}
The paper is organized as follows. In Section~\ref{sec:problem}, we state the problem of the Gaussian broadcast channel with receiver message side information and quickly review the lattice index coding scheme in \cite{viterbo14index}. In Section~\ref{sec:proposed_LIC}, we present the proposed lattice index coeds from number fields and prove some properties for the proposed codes. Some interesting design examples are given in Section~\ref{sec:design} followed by simulation results of LDPC coded index modulation in Section~\ref{sec:CIM}. Section~\ref{sec:conclude} concludes the paper. In Appendix~\ref{apx:prelim}, some background on algebra and algebraic number theory are reviewed. On the other hand, lattices and lattice codes have recently been popular and frequently discussed in the literature. Therefore, for the sake of brevity, we do not review lattices and lattice codes in this paper even though the proposed scheme requires the knowledge on them. The interested reader is referred to \cite{conway1999sphere} \cite{zamir_book} \cite{erez05}.

\subsection{Notations}
Throughout the paper, we use $\mbb{R}$ and $\mbb{C}$ to represent the set of real numbers and complex numbers, respectively. We use $j\defeq \sqrt{-1}$ to denote the imaginary unit. Vectors and matrices are written in lowercase boldface and uppercase boldface, respectively. We use $\times$ to denote the Cartesian product.

\section{Problem Statement}\label{sec:problem}
We consider that a sender broadcasting $K$ independent messages $(w_1,w_2,\ldots,w_K)$ to $L$ receivers. The sender jointly encodes the messages to $\mathbf{x}$ which belongs to a code (or constellation) $\mc{C}\subset\mbb{R}^n$ and is subject to the unit average power constraint. The received signal at the receiver $l$ is given by
\begin{equation}
    \mathbf{y}_l = \sqrt{\mathrm{SNR}_l} \mathbf{x} + \mathbf{z}_l,
\end{equation}
where $\mathbf{z}_l \sim \mc{N}(0,\mathbf{I}/n)$ is i.i.d. Gaussian noise and $\mathrm{SNR}_l$ represents the signal-to-noise ratio corresponding to the $l$th receiver. We further assume that the receiver $l$ has a subset of messages $w_{\mc{S}_l}\defeq \{w_k: k\in\mc{S}_l\}$ as side information which is governed by its index $\mc{S}_l \subset \{1,2,\ldots,K\}$ as shown in Fig.~\ref{fig:system_model}. This problem is an analogy to the index coding problem \cite{birk98} \cite{bar-yossef11} \cite{rouayheb10} extended to the physical-layer AWGN channel.  We shall refer to this model as the AWGN network.

In this paper, similar to \cite{viterbo14index}, we particularly study the case where all the $L$ receivers demand all the $K$ messages but can have individual subsets of side information $S_l$ for $l\in\{1,\ldots,L\}$. Let $R_k$ be the message rate of $w_k$, $R_{\mc{S}_l} = \sum_{k\in\mc{S}_l} R_k$, and $R=\sum_{k=1}^K R_k$. For the AWGN network, the capacity region of this special case can be derived from a result in \cite{tuncel06} and is given by
\begin{equation}
    R - R_{\mc{S}_l} < \frac{1}{2}\log(1+\mathrm{SNR}_l),\quad\text{bits/dim}\quad l\in\{1,\ldots,L\}.
\end{equation}
An interpretation of this region is that for a capacity-achieving code, every bit of side information provides roughly 6 dB of SNR reduction asymptotically. According to this intuition, \cite{viterbo14index} defines the side information gain of the receiver $l$ for a code $\mc{C}$ as
\begin{equation}
    \Gamma(\mc{C},\mc{S}_l) \defeq \frac{10\log_{10}(d^2_{\mc{S}_l}/d^2_0)}{R_{\mc{S}_l}}\quad \text{dB/bit/dim},
\end{equation}
where $d_0$ and $d_{\mc{S}_l}$ are the minimum Euclidean distance of the code before and after $\mc{S}_l$ is given, respectively. The side information gain essentially measures the gain in squared Euclidean distance provided by each bit of side information in $\mc{S}_l$, which is highly related to SNR reduction for a given probability of error. One can also define the overall side information gain for the system as
\begin{equation}
    \Gamma(\mc{C}) \defeq \min_{\mc{S}_l}\Gamma(\mc{C},\mc{S}_l)\quad \text{dB/bit/dim}.
\end{equation}


Natarajan \textit{et al.} then proposed a novel coding scheme which uses nested lattice codes from lattices over some principal ideal domains (PIDs) and termed this scheme as ``lattice index codes". In what follows, we briefly summarize the lattice index code in \cite{viterbo14index}. We particularly use $\mbb{Z}$ for example for the sake of simplicity; however, designs over $\Zi$, $\Zw$, and $\mbb{H}$ are available in \cite{viterbo14index} as well.

Let $p_1,\ldots,p_K$ be $K$ distinct primes and $q=\Pi_{k=1}^K p_k$. We start with the one-dimensional case where lattice index codes are actually lattice index modulation. The lattice index code in \cite{viterbo14index} enforces $w_k\in\mbb{F}_{p_k}$ and encodes the message $(w_1,\ldots,w_K)$ to
\begin{equation}\label{eqn:lim_viter}
    x = \gamma\cdot\left[\left(\frac{q}{p_1} w_1 + \ldots + \frac{q}{p_K} w_K \right)\hspace{-3pt}\mod q\mbb{Z}\right],
\end{equation}
where $\gamma\in\mbb{R}$ is for satisfying the power constraint. One observation we would like to point out (and discuss in more detail later) is that the mapping used above is closely related to an isomorphism guaranteed by the Chinese remainder theorem (CRT). It has been shown in \cite{viterbo14index} that lattice index codes thus constructed provide uniform side information gains of 6 dB. i.e., $\Gamma(\mc{C},\mc{S})$ is 6 dB for every choice of $\mc{S}\subset \{1,2,\ldots,K\}$. This result shows that the lattice index code in \cite{viterbo14index} mimics the behavior of a capacity-achieving code. Moreover, any side information gain larger than 6 dB translates into packing inefficiency of the original code $\mc{C}$ when used over the AWGN network \cite{viterbo14index}.

Now, let $\Lambda$ be a $\mbb{Z}$-lattice. Let $\Lambda_s = q\Lambda$ and $\Lambda_k = \frac{q}{p_k}\Lambda$. A lattice index code (with dimension larger than 1), is given by
\begin{equation}\label{eqn:lic_viter}
    \mathbf{x} = \gamma\cdot\left[\left(\Lambda_1 + \ldots + \Lambda_K \right)\hspace{-3pt}\mod \Lambda_s\right],
\end{equation}
where the message $w_k$ is encoded by $\Lambda_k/\Lambda_s$ and $\gamma\in\mbb{R}$ is again for the power constraint. Natarajan \textit{et al.} then showed that lattice index codes constructed by the above procedure with the densest $\Lambda$ in that dimension can provide an uniform side information gain of 6 dB.

\begin{remark}\label{rmk:CIM_vs_self_similar}
    We point out that the study of the one-dimensional case in \eqref{eqn:lim_viter} is very interesting despite the fact that it cannot provide coding gains. For a good lattice index modulation that provides large side information gains, one can empower it by concatenating it with powerful linear outer codes to obtain large coding gains as in \cite{viterbo_isit15}. This separation philosophy is in general prevailing in current communication systems and will be discussed in Section~\ref{sec:CIM}. Moreover, the lattice index code in \eqref{eqn:lic_viter} is restricted to self-similar lattices (all lattices are self-similar to $\Lambda$) and hence limits the design space. This may bring difficulties to rate allocation. However, for the sake of completeness, we still provide a generalization of the lattice index coding scheme in \eqref{eqn:lic_viter} to number fields in Appendix~\ref{apx:lic_NF}.
\end{remark}

Another model that will be considered in this paper models the case where we have both Rayleigh fading and AWGN noise. This model is called the Rayleigh fading network and is given by
\begin{equation}
    \mathbf{y}_l = \sqrt{\mathrm{SNR}_l} \mathbf{H}_l \mathbf{x} + \mathbf{z}_l,
\end{equation}
where $\mathbf{H}_l$ is a diagonal matrix with diagonal elements $h_{l1}, \ldots, h_{ln}$ drawn from i.i.d. Rayleigh distribution. We further assume that the receiver $l$ has perfect knowledge of $\mathbf{H}_l$. It is well known that (see for example \cite{joseph96} \cite{tse_viswananth} \cite{oggier04}) for such a setting, diversity order $D(\mc{C})$ and the $D(\mc{C})$-minimum product distance of $\mc{C}$ are important performance metrics. The former is the minimum number of entries that $\mathbf{x}_1\neq \mathbf{x}_2\in\mc{C}$ differ from each other and the latter is defined when every pair of $\mathbf{x}_1$ and $\mathbf{x}_2$ has at least $D(\mc{C})$ different components as
\begin{equation}
    d_{p,min}(\mc{C})\defeq\underset{\mathbf{x}_1,\mathbf{x}_2\in\mc{C}}{\min} \Pi_{x_{1i}\neq x_{2i}} |x_{1i} - x_{2i}|.
\end{equation}
After revealing the messages in $\mc{S}_l$ to the receiver $l$, one can similarly define the diversity order $D(\mc{C},\mc{S}_l)$ and the $D(\mc{C},\mc{S}_l)$-minimum product distance $d_{p,min}(\mc{C},\mc{S}_l)$.


\section{Lattice Index Codes from Number Fields}\label{sec:proposed_LIC}
In this section, we propose construction of lattice index codes over rings of algebraic integers. This can be regarded as an extension of the scheme in \cite{viterbo14index} to number fields. We first discuss the proposed construction and then show some important properties of the proposed scheme. The proposed scheme heavily uses abstract algebra and algebraic number theory that are briefly reviewed in Appendix~\ref{apx:prelim}.

\subsection{Construction over Rings of Algebraic Integers}\label{sec:proposed_const}
The scheme proposed in this section is based on the observation that the lattice index code in \eqref{eqn:lim_viter} is closely related to
\begin{align}
    x &= \gamma\cdot\left[\mc{M}(w_1,\ldots,w_K)\hspace{-3pt}\mod q\mbb{Z}\right] \nonumber \\
    &= \gamma\cdot\left[\left(v_1\frac{q}{p_1} w_1 + \ldots + v_K\frac{q}{p_K} w_K \right)\hspace{-3pt}\mod q\mbb{Z}\right],
\end{align}
where $\mc{M}:\mbb{F}_{p_1}\times\ldots\times\mbb{F}_{p_K}\rightarrow \mbb{Z}/q\mbb{Z}$ is an isomorphism guaranteed by CRT and $v_k\in\mbb{Z}$ is the coefficient of the B\'{e}zout identity for $k\in\{1,\ldots,K\}$. In what follows, we generalize this idea to a general ring of algebraic integers and prove some properties of the proposed scheme. The main challenge lies in the fact that a general ring of algebraic integers may not form a PID and hence we cannot equivalently work with numbers as we did for schemes over a PID.

Let $\mbb{K}$ be an algebraic number field with degree $n=[\mbb{K}:\mbb{Q}]$ and signature $(r_1,r_2)$ and let $\Ok$ be its ring of integers. Let $\mfk{p}_1,\ldots,\mfk{p}_K$ be prime ideals of $\Ok$ lying above $p_1,\ldots,p_K$ with inertial degrees $f_1,\ldots,f_K$, respectively. Moreover, $\mfk{p}_1,\ldots,\mfk{p}_K$ are relatively prime. From CRT, there exists a ring isomorphism
\begin{equation}
    \mc{M}: \mbb{F}_{p_1^{f_1}}\times\ldots\times\mbb{F}_{p_K^{f_K}}\rightarrow \Ok/\Pi_{k=1}^{K}\mfk{p}_k.
\end{equation}
Hence, one can have the decomposition
\begin{equation}\label{eqn:decompose}
    \Ok=\mc{M}(\mbb{F}_{p_1^{f_1}},\ldots,\mbb{F}_{p_K^{f_K}})+\Pi_{k=1}^K\mfk{p}_k.
\end{equation}
One can then map everything to a lattice over $\mbb{R}^{r_1}\times\mbb{C}^{r_2}\cong\mbb{R}^n$  by $\Psi(.)$ the canonical embedding defined in \eqref{eqn:embedding} in Appendix~\ref{apx:prelim} to get $\Lambda_{\Ok}=\mc{M}(\mbb{F}_{p_1^{f_1}},\ldots,\mbb{F}_{p_K^{f_K}})+\Lambda_{\Pi_{k=1}^K\mfk{p}_k}$. Note that here, we abuse the notation to use the same $\mc{M}$ to denote the ring isomorphism before and after the mapping $\Psi$. 

We are now ready to describe the proposed lattice index coding scheme which essentially follows the same idea in \cite{viterbo14index}. Let $w_k\in\mbb{F}_{p_k^{f_k}}$ for $k\in\{1,\ldots,K\}$. For a particular set of $(w_1,\ldots,w_K)\in \mbb{F}_{p_1^{f_1}} \times \ldots \times \mbb{F}_{p_K^{f_K}}$, the proposed lattice index coding scheme first forms a complex number
\begin{equation}\label{eqn:proposed_LIC}
    \tilde{x} = \mc{M}(w_1,\ldots,w_K)\hspace{-3pt}\mod \Pi_{k=1}^K \mfk{p}_k.
\end{equation}
We then transmit $\mathbf{x}=\gamma\Psi(\tilde{x})$ where $\gamma\in\mbb{R}$ is for the power constraint. Note that the constellation defined above can be any complete set of coset representatives. We particularly choose one with the minimum energy as our constellation for saving energy. Some examples are provided in the following.

\begin{example}[PID]\label{exp:real_PID}
Consider $\mbb{K}=\mbb{Q}(\sqrt{5})$ whose ring of integers is $\mbb{Z}\left[(1+\sqrt{5})/2\right]$. Note that this is a PID so we only have to deal with numbers as every ideal is generated by a singleton. Consider the following two prime numbers $\phi_1 = \sqrt{5}$ and $\phi_2 = 4+\sqrt{5}$ with $N(\phi_1)=5$ and $N(\phi_2)=11$, respectively. i.e., $\phi_1\Ok$ and $\phi_2\Ok$ lie above $5$ and $11$, respectively. The product of these two ideals is exactly $\phi_1\phi_2\Ok$. Note that for this ring, $\Psi(\tilde{x})=[\sigma_1(\tilde{x}),~\sigma_2(\tilde{x})]$ where $\sigma_1(a+b\sqrt{5})=a+b\sqrt{5}$ and $\sigma_2(a+b\sqrt{5})=a-b\sqrt{5}$. Fig.~\ref{fig:Qr5_given_1} shows the proposed lattice index coding scheme constructed with these ideals. The labels in this figure represent the mapping from $\mbb{F}_5\times\mbb{F}_{11}$ to the constellation. One can verify that this mapping is a ring isomorphism from $\mbb{F}_5\times\mbb{F}_{11}$ to $\Ok/\phi_1\phi_2\Ok$. Also, the triangles in this figure represent the constellation after the first message is given to be 0, i,e, $w_1=0$. The same lattice index coding scheme is shown in Fig.~\ref{fig:Qr5_given_2} in which we also denote by stars the constellation after the second message is given to be 0, i.e., $w_2=0$.

\begin{figure}
    \centering
    \includegraphics[width=3.5in]{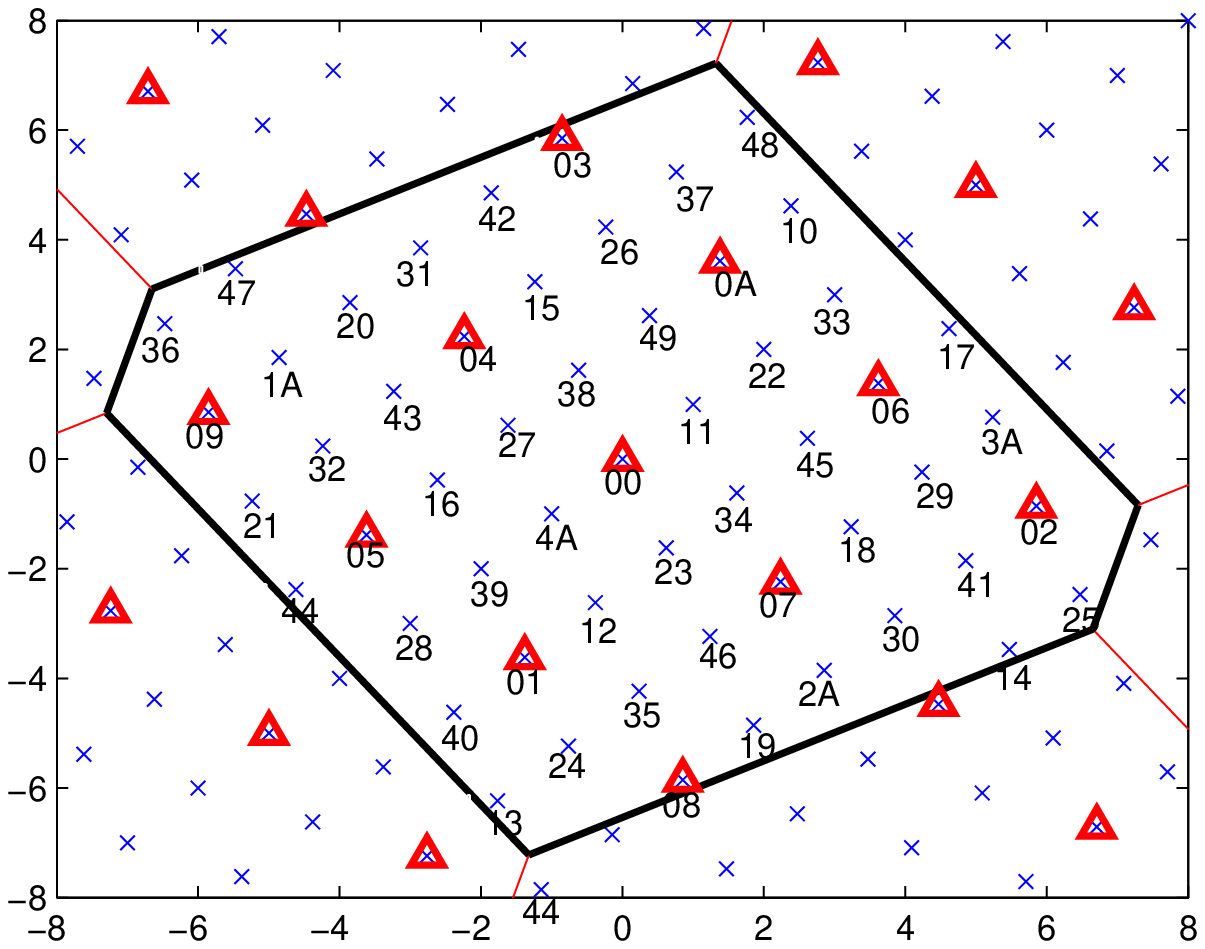}
    \caption{A lattice index constellation constructed from $\mbb{Q}(\sqrt{5})$ where only those lattice points inside the black box are used and the labels are the mapping (a ring isomorphism). The triangles form the constellation after the first message is given to be 0.}
    \label{fig:Qr5_given_1}
\end{figure}
\begin{figure}
    \centering
    \includegraphics[width=3.5in]{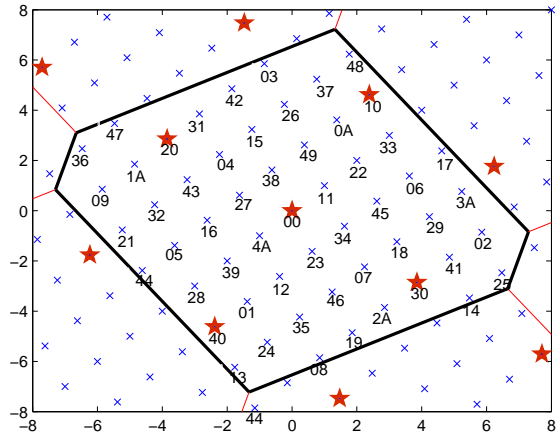}
    \caption{A lattice index constellation constructed from $\mbb{Q}(\sqrt{5})$ where only those lattice points inside the black box are used and the labels are the mapping (a ring isomorphism). The stars form the constellation after the second message is given to be 0.}
    \label{fig:Qr5_given_2}
\end{figure}
\end{example}

\begin{example}[Non-PID]\label{exp:noPID}
Consider $\mbb{K}=\mbb{Q}(\sqrt{-5})$ whose ring of integers is $\mbb{Z}[\sqrt{-5}]$. Note that this is not a PID so we have to work with ideals. Consider the following two prime ideals $\mfk{p}_1=(7,3+\sqrt{-5})$ and $\mfk{p}_2 = \bar{\mfk{p}}_1$ both lying above $7$, respectively. Note that $\mfk{p}_1\mfk{p}_2=7\Ok$ is a principal ideal. Fig.~\ref{fig:Q_n5_49_given_1} shows the proposed lattice index coding scheme thus constructed. The labels in this figure represent the mapping from $\mbb{F}_7\times\mbb{F}_7$ to the constellation. One can verify that this mapping is a ring isomorphism from $\mbb{F}_7\times\mbb{F}_7$ to $\Ok/\mfk{p}_1\cdot\mfk{p}_2$. Also, the triangles in this figure represent the constellation after fixing $w_1=0$. The same lattice index coding scheme is shown in Fig.~\ref{fig:Q_n5_49_given_2} in which we also denote by stars the constellation after fixing $w_2=0$.

\begin{figure}
    \centering
    \includegraphics[width=3.5in]{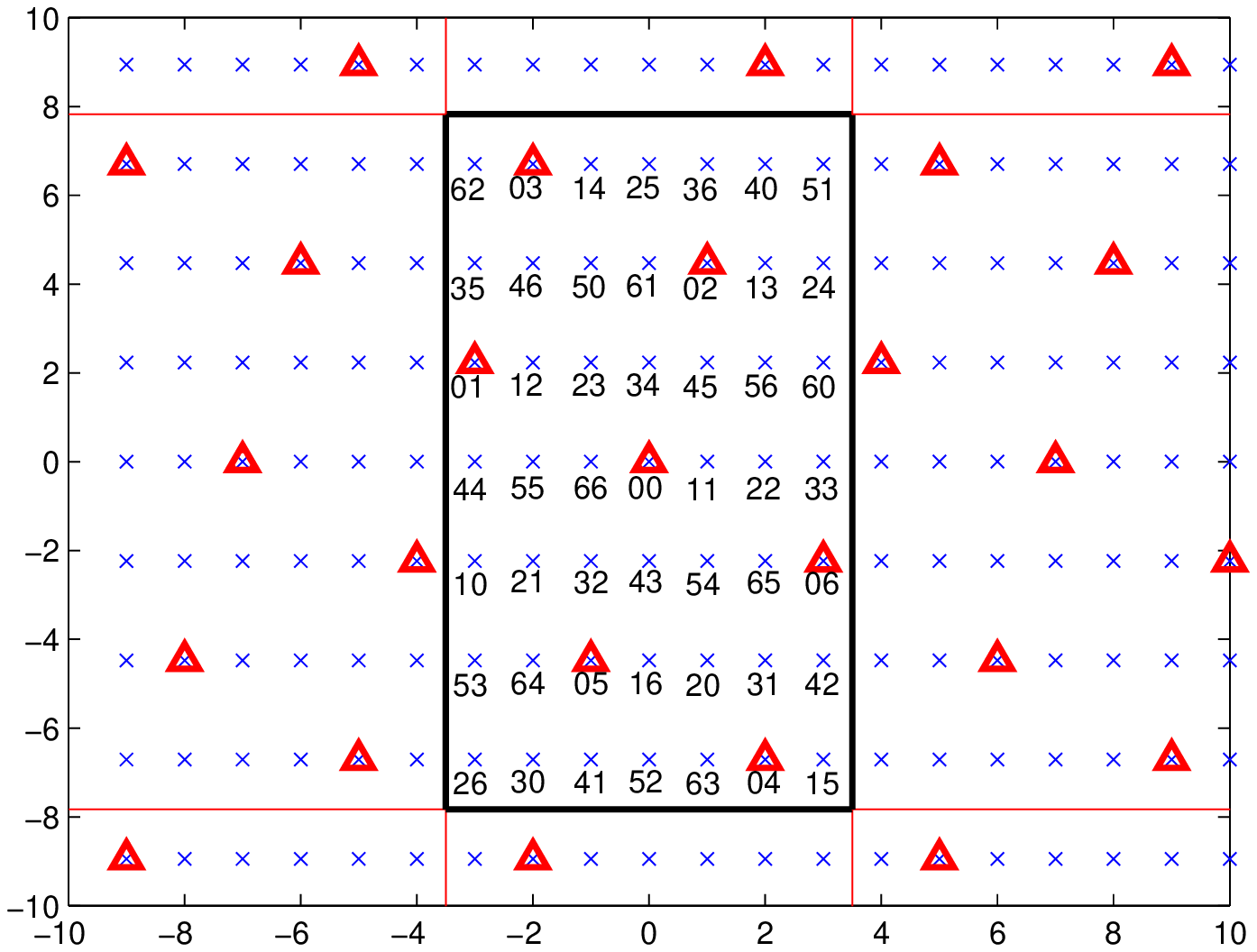}
    \caption{A lattice index constellation constructed from $\mbb{Q}(\sqrt{-5})$ where only those lattice points inside the black box are used and the labels are the mapping (a ring isomorphism). The triangles form the constellation after the first message is given to be 0.}
    \label{fig:Q_n5_49_given_1}
\end{figure}
\begin{figure}
    \centering
    \includegraphics[width=3.5in]{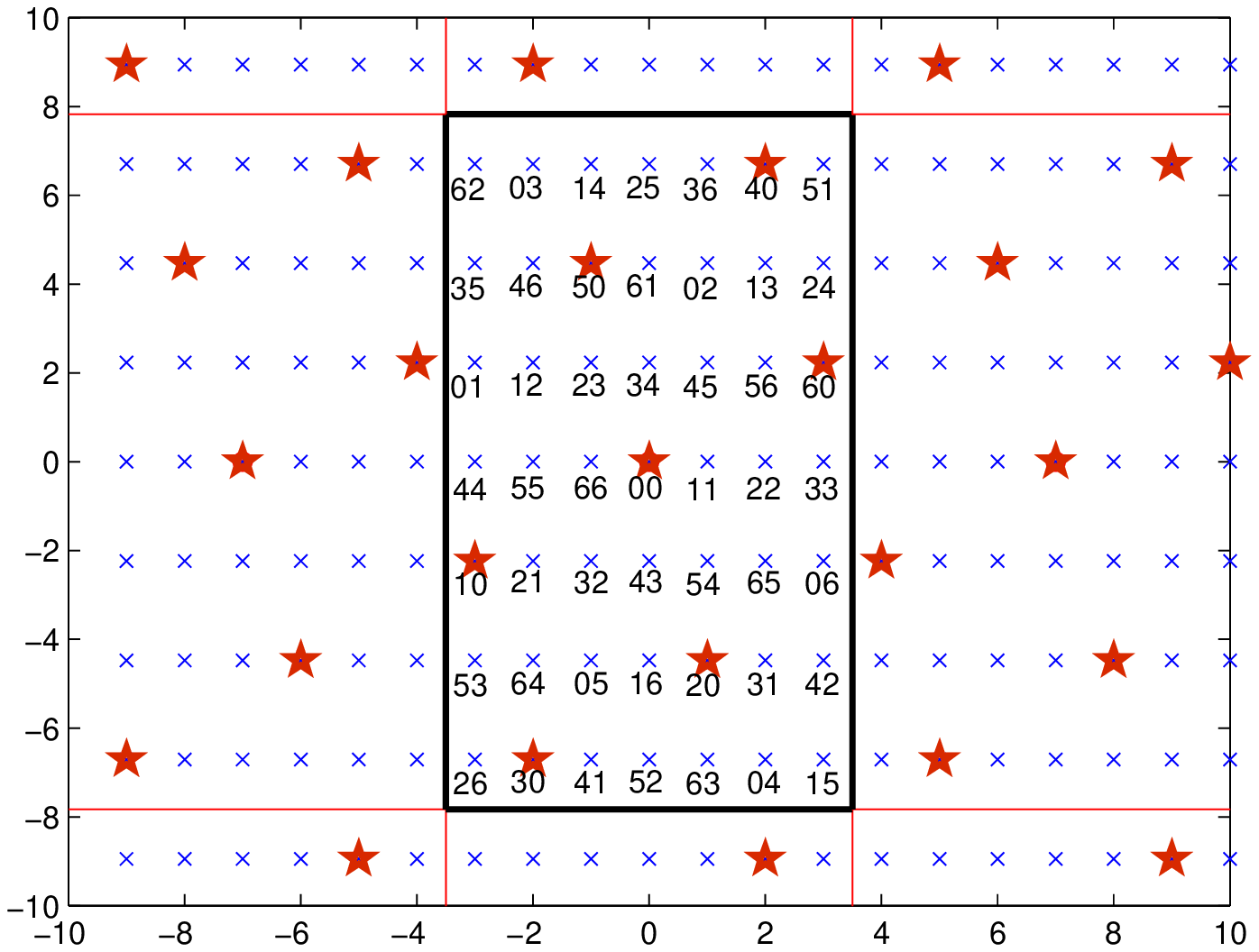}
    \caption{A lattice index constellation constructed from $\mbb{Q}(\sqrt{-5})$ where only those lattice points inside the black box are used and the labels are the mapping (a ring isomorphism). The stars form the constellation after the second message is given to be 0.}
    \label{fig:Q_n5_49_given_2}
\end{figure}
\end{example}

\begin{example}[PID]\label{exp:PID}
Consider $\mbb{K}=\mbb{Q}(\sqrt{-7})$ whose ring of integers is $\mbb{Z}\left[(1+\sqrt{-7})/2\right]$. Note that this is a PID so we only have to deal with numbers. Consider the following two prime numbers $\phi_1 = \sqrt{-7}$ and $\phi_2 = 2+\sqrt{-7}$ such that $N(\phi_1)=7$ and $N(\phi_2)=11$, respectively. i.e., $\phi_1\Ok$ and $\phi_2\Ok$ lie above $7$ and $11$, respectively. The product of these two ideals is exactly $\phi_1\phi_2\Ok$. Note that for this ring, $\Psi(\tilde{x})=\sigma_1(\tilde{x})= \tilde{x}\in\mbb{C}$. Fig.~\ref{fig:Q_n7_given_7} shows the proposed lattice index coding scheme constructed with these ideals. The labels in this figure represent the mapping from $\mbb{F}_7\times\mbb{F}_{11}$ to the constellation. One can verify that this mapping is a ring isomorphism from $\mbb{F}_7\times\mbb{F}_{11}$ to $\Ok/\phi_1\phi_2\Ok$. Also, the triangles in this figure represent the constellation after fixing $w_1=0$. The same lattice index coding scheme is shown in Fig.~\ref{fig:Q_n7_given_11} in which we also denote by stars the constellation after fixing $w_2=0$.

\begin{figure}
    \centering
    \includegraphics[width=3.5in]{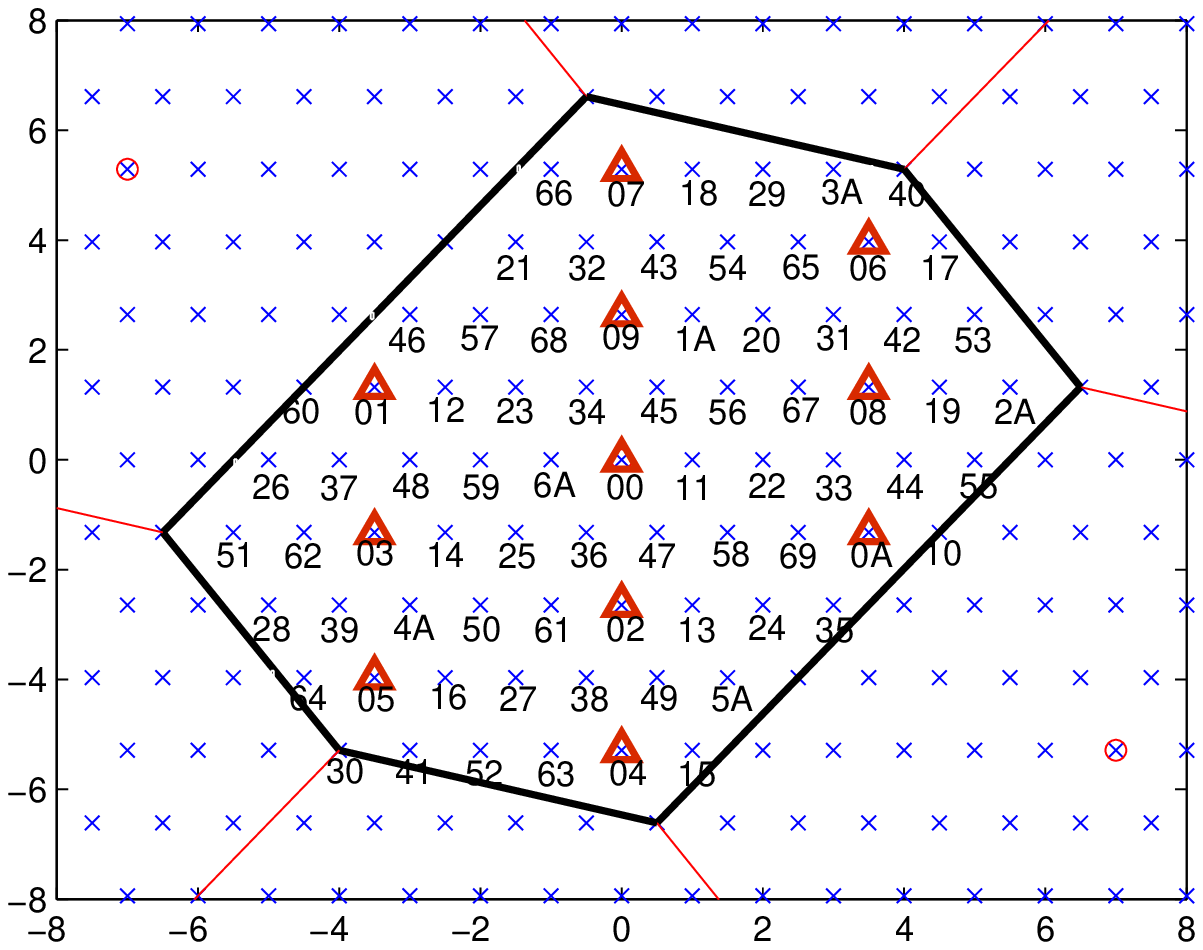}
    \caption{A lattice index constellation constructed from $\mbb{Q}(\sqrt{-7})$ where only those lattice points inside the black box are used and the labels are the mapping (a ring isomorphism). The triangles form the constellation after the first message is given to be 0.}
    \label{fig:Q_n7_given_7}
\end{figure}
\begin{figure}
    \centering
    \includegraphics[width=3.5in]{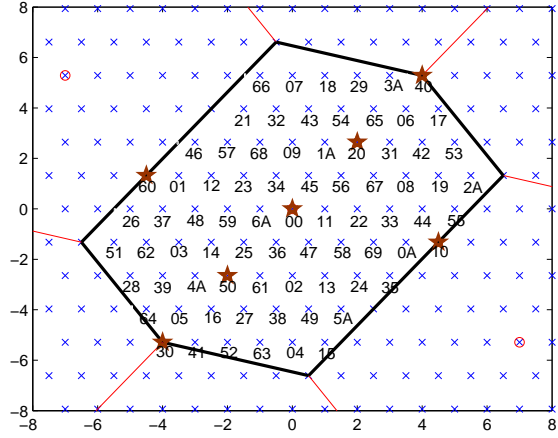}
    \caption{A lattice index constellation constructed from $\mbb{Q}(\sqrt{-7})$ where only those lattice points inside the black box are used and the labels are the mapping (a ring isomorphism). The stars form the constellation after the second message is given to be 0.}
    \label{fig:Q_n7_given_11}
\end{figure}
\end{example}

\begin{remark}
    Here, we always choose $\mfk{p}_1,\ldots,\mfk{p}_K$ to be prime ideals that are relatively prime, which is by no means necessary. In fact, the CRT only requires those ideals to be relatively prime. i.e., $\mfk{p}_k+\mfk{p}_{k'}=\Ok$ for any pair of $(k,k')\in\{1,\ldots,K\}^2$ in order to make $\cap_{k=1}^K \mfk{p}_k = \Pi_{k=1}^K \mfk{p}_k$. The reason that we restrict ourselves to prime ideals is because in a typical communication system, messages are usually encoded by some codes over finite fields. Therefore, it is of primary interest to study the case where we actually get a field. Nonetheless, all the results are generlizable to the general case where $\mfk{p}_k$s are relatively prime but may not be prime ideals.
\end{remark}

\subsection{Properties of the proposed lattice index codes}\label{sec:proposed_property}
We now provide some properties of the proposed scheme about the side information gains. Before proceeding, we note that since scaling by a scalar would not change the structure of a lattice, we set $\gamma=1$ from now on unless otherwise specified.

Consider a generic index set $\mc{S}$, let $\mfk{I}=\Pi_{k\in\mc{S}} \mfk{p}_k$. The message rates $R_{\mc{S}}$ can then be computed from the coset decomposition \eqref{eqn:decompose} as
\begin{align}\label{eqn:R_s}
    R_{\mc{S}} &= \frac{1}{n}\log_2\frac{\text{Vol}(\Lambda_{\mfk{I}})}{\text{Vol}(\Lambda_{\Ok})}= \frac{1}{n}\log_2|\Ok/\mfk{I}|  \nonumber \\
        &\overset{(a)}{=} \frac{1}{n}\log_2 N(\mfk{I}) \overset{(b)}{=} \frac{1}{n}\sum_{k\in\mc{S}}\log_2 N(\mfk{p}_k)\quad \text{bits/dim.},
\end{align}
where (a) follows from the definition of the ideal norm operation and (b) is because the ideal norm operation is multiplicative. Moreover, notice that suppose $w_1=v_1$ is given, the signal before embedding becomes
    \begin{align}
        \tilde{x}_1 &= \mc{M}(v_1,w_2,\ldots,w_K)\hspace{-3pt}\mod \Pi_{k=1}^K \mfk{p}_k \nonumber \\
        &\overset{(a)}{=} \mc{M}(v_1,0,\ldots,0) + \mc{M}(0,w_2,\ldots,w_K)\hspace{-3pt}\mod \Pi_{k=1}^K \mfk{p}_k,
    \end{align}
where (a) is because $\mc{M}$ is a ring isomorphism. One can then see that $\tilde{x}_1$ is a shifted version of $\tilde{x}_1' = \mc{M}(0,w_2,\ldots,w_K)\hspace{-3pt}\mod \Pi_{k=1}^K \mfk{p}_k$. As a consequence of CRT, one has $\tilde{x}_1'\in \mfk{p}_1$. For a general set $\mc{S}\subset \{1,\ldots,K\}$, we define $\tilde{x}'_{\mc{S}}$ to be the signal obtained by setting $w_k=0$ for $k\in\mc{S}$ in $x$. It follows that CRT guarantees $\tilde{x}'_{\mc{S}}\in \Pi_{k\in\mc{S}} \mfk{p}_k$. And hence, $\mathbf{x}$ is a shifted version of $\Psi(\tilde{x}'_{\mc{S}})\in \Psi(\Pi_{k\in\mc{S}} \mfk{p}_k)=\Psi(\mfk{I})$. In the following, we can without loss of generality assume that $\mathbf{x}=\Psi(\tilde{x}'_{\mc{S}})$ as shifting by a constant will not change the lattice structure.

We now use the Minkowski theorem in geometry of numbers to show an upper bound on the minimum distance of this constellation. Note that this method has been used in \cite[Sec. 6]{peikert07}) for ideal lattices obtained by a different embedding. Note that the fundamental Voronoi region of the lattice $\Psi(\mfk{I})$ is given by
    \begin{equation}
        \text{Vol}(\mc{V}_{\Psi(\mfk{I})}) = \frac{\sqrt{|\Delta_{\mbb{K}}|}}{2^{r_2}}N(\mfk{I}).
    \end{equation}
    Let $\mc{B}$ be a closed convex body lying in $\mbb{R}^{r_1}\times\mbb{C}^{r_2}$ as follows,
    \begin{equation}
        \mc{B} \defeq\left\{ \mathbf{b}\in \mbb{R}^{r_1}\times\mbb{C}^{r_2} : \|\mathbf{b}\|_{\infty}\leq 1 \right\},
    \end{equation}
    where $\|\mathbf{b}\|_{\infty}$ is the infinity norm. The volume of $\mc{B}$ is
    \begin{equation}
        \text{Vol}(\mc{B}) = 2^{r_1} \cdot \pi^{r_2}.
    \end{equation}
    Now, let
    \begin{equation}
        \beta = \left(\sqrt{|\Delta_{\mbb{K}}|}N(\mfk{I})\right)^{\frac{1}{n}}\left(\frac{2}{\pi}\right)^{\frac{r_2}{n}}.
    \end{equation}
    We have
    \begin{equation}
        \text{Vol}(\beta\mc{B}) = 2^n\frac{\sqrt{\Delta_{\mbb{K}}}}{2^{r_2}}N(\mfk{I}) = 2^n \text{Vol}(\mc{V}_{\Psi(\mfk{I})}).
    \end{equation}
    Hence, from the Minkowski theorem in geometry of numbers, there must exist at least one non-zero lattice point $\boldsymbol\lambda$ inside $\beta\mc{B}$. This implies that the minimum distance of this constellation is upper-bounded as follows,
    \begin{align}\label{eqn:d_s_bound}
        d_{\mc{S}}&\leq\|\boldsymbol\lambda\|  \leq \sqrt{r_1+r_2}\cdot\|\boldsymbol\lambda\|_{\infty}
        \leq \sqrt{r_1+r_2}\cdot\beta \nonumber \\
        &\leq \sqrt{r_1+r_2}\left(\sqrt{|\Delta_{\mbb{K}}|}N(\mfk{I})\right)^{\frac{1}{n}}\left(\frac{2}{\pi}\right)^{\frac{r_2}{n}}.
    \end{align}

In what follows, we use the above properties to prove bounds on the side information gains for the proposed lattice index codes from some particular families of algebraic fields. The first one concerns totally real number fields.
%
%
\begin{theorem}\label{thm:si_gain_real}
    For the proposed lattice coding scheme over a \textit{totally real} number field $\mbb{K}$ with discriminant $\Delta_{\mbb{K}}$, the side information gain provided by $\mc{S}\subset\{1,\ldots,K\}$ can be bounded as follows.
    \begin{equation}
        6 \leq \Gamma(\mc{C},\mc{S}) \leq 6 + \gamma_{\mc{S}}\quad \text{dB/bit/dim},
    \end{equation}
    where $\gamma_{\mc{S}}\defeq \frac{10\log_{10}|\Delta_{\mbb{K}}|}{\sum_{k\in\mc{S}}f_k\log_2 p_k}$. Moreover, when used over the Rayleigh fading network, this scheme provides diversity order $D(\mc{C})=n$ and $d_{p,min}(\mc{C})= 1$. Also, for any $\mc{S}\subset\{1,\ldots,K\}$, we have $D(\mc{C},\mc{S})=n$ and $d_{p,min}(\mc{C},\mc{S})\geq \Pi_{k\in\mc{S}} p_k^{f_k}$ for any $\mc{S}$.
\end{theorem}
\begin{proof}
    For a totally real number field $\mbb{K}$, it has the signature $(r_1,r_2)=(n,0)$. Hence, from \eqref{eqn:d_s_bound}, an upper bound on $d_{\mc{S}}$ can be easily shown by plugging in $r_1=n$ and $r_2=0$. In what follows, we use the AM-GM inequality to prove a lower bound on $d_{\mc{S}}$.

    Again, for a general set $\mc{S}\subset \{1,\ldots,K\}$,  CRT guarantees that $\tilde{x}'_{\mc{S}}\in \Pi_{k\in\mc{S}} \mfk{p}_k=\mfk{I}$. For a $\tilde{x}\in\mfk{I}$, after embedding, the squared distance from 0 (it suffices to consider the distance from $\mathbf{0}$ due to the lattice structure) can be bounded as follows,
    \begin{align}
        \|\Psi(\tilde{x})\|^2 &= \sum_{i=1}^n \sigma_i(\tilde{x})^2 \overset{(a)}{\geq} n\left(\Pi_{i=1}^n \sigma_i(\tilde{x})\right)^{\frac{2}{n}} \nonumber \\
        &= nN_{\mbb{K}}(\tilde{x})^{\frac{2}{n}} \geq nN(\mfk{I})^{\frac{2}{n}},
    \end{align}
    where (a) follows from the AM-GM inequality. Therefore, the minimum distance of this constellation can be lower-bounded by
    \begin{equation}\label{eqn:d_s_lower_bound}
        d_{\mc{S}} \geq \sqrt{n}N(\mfk{I})^{1/n}.
    \end{equation}
    Moreover, the above bound also indicates that $d_0\geq\sqrt{n}N(\Ok)^{1/n}=\sqrt{n}$. Combining \eqref{eqn:R_s}, \eqref{eqn:d_s_bound}, \eqref{eqn:d_s_lower_bound}, and the fact that $d_0=\sqrt{n}$ (since $1\in\Ok$) results in
    \begin{equation}\label{eqn:bound_gamma}
        \frac{20\log_{10}N(\mfk{I})}{\log_2 N(\mfk{I})}\leq \Gamma(\mc{C},\mc{S}) \leq \frac{20\log_{10}N(\mfk{I}) + 10\log_{10}|\Delta_{\mbb{K}}|}{\log_2 N(\mfk{I})} ,
    \end{equation}
    and thus
    \begin{equation}
        6\leq \Gamma(\mc{C},\mc{S})\leq 6 + \gamma_{\mc{S}} \quad \text{dB/bit/dim}.
    \end{equation}

    Now, let us consider using this scheme over the Rayleigh fading network. We first note that for the proposed scheme, $\tilde{x}\in\Ok$ and $\mathbf{x}\in\Psi(\tilde{x})$; thus, $\mc{C}\subseteq\Psi(\Ok)$ will have diversity order $D(\mc{C})=n$ for a totally real $\Ok$ \cite{joseph96}. The $n$-product distance of $\mathbf{x}$ from $\mathbf{0}$ is lower-bounded by
    \begin{align}
        d_{p}(\mathbf{x},\mathbf{0}) &= d_{p}(\Psi(\tilde{x}),\mathbf{0})= \Pi_{i=1}^n |\sigma_i(\tilde{x})| \nonumber \\
        &= |N_{\mbb{K}}(\tilde{x})| \geq N(\Ok) \nonumber \\
        &= 1.
    \end{align}
    This is true for any $\mathbf{x}\in\mc{C}$; thus, together with the fact that $d_p(\mathbf{1},\mathbf{0})=1$, we conclude that $d_{p,min}(\mc{C})= 1$.

    For a $\mc{S}\subset\{1,\ldots,K\}$, recall that when given $\mc{S}$, $\tilde{x}'_{\mc{S}}$ belongs to $\Pi_{k\in\mc{S}}\mfk{p}_k$ an ideal of $\Ok$; therefore, the diversity order of $\Psi(\Pi_{k\in\mc{S}}\mfk{p}_k)$ and that of $\Psi(\Ok)$ are the same. i.e., $D(\mc{C},\mc{S})=n$ for a totally real $\Ok$. Moreover, the $n$-product distance of $\mathbf{x}$ from $\mathbf{0}$ is lower-bounded by
    \begin{align}
        d_{p}(\mathbf{x},\mathbf{0}) &= d_{p}(\Psi(\tilde{x}),\mathbf{0})= \Pi_{i=1}^n |\sigma_i(\tilde{x})| \nonumber \\
        &= |N_{\mbb{K}}(\tilde{x})| \geq N(\Pi_{k\in\mc{S}}\mfk{p}_k) \nonumber \\
        &= \Pi_{k\in\mc{S}} p_k^{f_k}.
    \end{align}
    The above bound is true for any $\mathbf{x}\in\mc{C}$ with $w_{\mc{S}}$ fixed; therefore, $d_{p,min}(\mc{C},\mc{S})\geq \Pi_{k\in\mc{S}} p_k^{f_k}$.
\end{proof}
We now provide simulation results with the proposed lattice index code in Example~\ref{exp:real_PID}. We first use this scheme over the AWGN network as shown in Fig~\ref{fig:Qr5_fig_AWGN} where we observe a side information gain of 7 dB and 11 dB at symbol error rate of $10^{-5}$ when revealing $w_1$ and $w_2$ to the receiver, respectively. These correspond to $7/\frac{1}{2}\log_2(5)\approx 6.03$ and $11/\frac{1}{2}\log_2(11)\approx 6.36$ dB/bit/dim where the bounds above are $6\leq \Gamma(\mc{C},\{1\}) \leq 9.01$ and $6\leq \Gamma(\mc{C},\{2\}) \leq 8.02$ dB/bit/dim for revealing $w_1$ and $w_2$, respectively.
\begin{figure}
    \centering
    \includegraphics[width=3.5in]{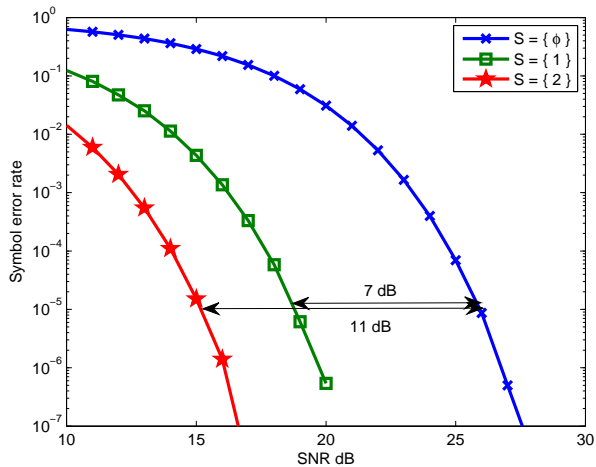}
    \caption{SNR versus symbol error rate over the AWGN network. The scheme is described in Example~\ref{exp:real_PID}.}
    \label{fig:Qr5_fig_AWGN}
\end{figure}

Theorem~\ref{thm:si_gain_real} also indicates that this scheme can provide a diversity gain of $D(\mc{C})=D(\mc{C},\mc{S})=2$ for any $\mc{S}\subset\{1,\ldots,K\}$ when used over the Rayleigh fading network. We examine this in Fig.~\ref{fig:Qr5_fig_ray} where we adopt the same scheme and use it over the Rayleigh fading network. One observes that as shown in Theorem~\ref{thm:si_gain_real}, the proposed lattice index code has diversity order 2 for any $\mc{S}$. Moreover, we point out that, compared with the results in the AWGN network, the side information gains increase from 7 dB to 8.5 dB for $S=\{1\}$ and from 11 dB to 13 dB for $S=\{2\}$. This is a consequence of having increased $n$-minimum product distances when having receiver message side information as shown in Theorem~\ref{thm:si_gain_real}.
\begin{figure}
    \centering
    \includegraphics[width=3.5in]{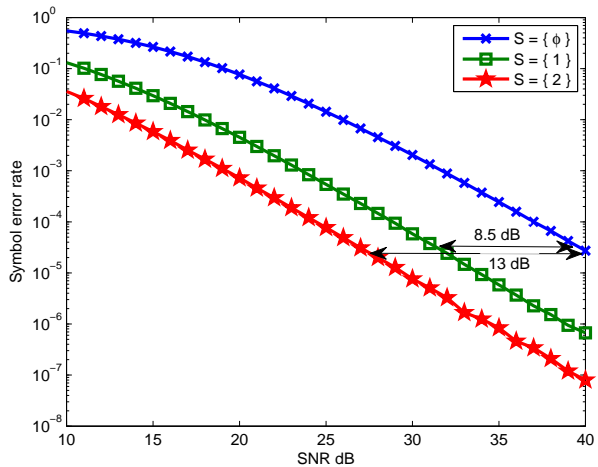}
    \caption{SNR versus symbol error rate over the Rayleigh fading network. The scheme is described in Example~\ref{exp:real_PID}.}
    \label{fig:Qr5_fig_ray}
\end{figure}

The following theorem considers totally complex number fields.
\begin{theorem}\label{thm:si_gain_complex}
    For the proposed lattice coding scheme over a \textit{totally complex} number field $\mbb{K}$ with discriminant $\Delta_{\mbb{K}}$, the side information gain provided by $\mc{S}\subset\{1,\ldots,K\}$ can be bounded as follows.
    \begin{equation}
        6 \leq \Gamma(\mc{C},\mc{S}) \leq 6 + \tilde{\gamma}_{\mc{S}} \quad \text{dB/bit/dim},
    \end{equation}
    where now $\tilde{\gamma}_{\mc{S}}\defeq \frac{10\log_{10}|\Delta_{\mbb{K}}|\left(\frac{2}{\pi}\right)^n}{\sum_{k\in\mc{S}}f_k\log_2 p_k}$. When used over the Rayleigh fading network, this scheme provides diversity order $D(\mc{C})=n/2$. Moreover, for any $\mc{S}\subset\{1,\ldots,K\}$, we have $D(\mc{C},\mc{S})=n/2$.
\end{theorem}
\begin{proof}
    For a totally complex number field $\mbb{K}$, it has the signature $(r_1,r_2)=(0,n/2)$. Hence, from \eqref{eqn:d_s_bound}, an upper bound on $d_{\mc{S}}$ can be easily shown by plugging in $r_1=0$ and $r_2=n/2$. In what follows, we again use the AM-GM inequality to prove a lower bound on $d_{\mc{S}}$.

    For a $x\in\mfk{I}$, after embedding, the squared distance from $\mathbf{0}$ can be bounded as follows
    \begin{align}
        \|\Psi(\tilde{x})\|^2 &= \sum_{i=1}^{n/2} |\sigma_i(\tilde{x})|^2 \overset{(a)}{=} \frac{1}{2}\sum_{i=1}^n|\sigma_i(\tilde{x})|^2 \nonumber \\
        &\overset{(b)}{\geq} \frac{n}{2}\left(\Pi_{i=1}^n |\sigma_i(\tilde{x})|\right)^{\frac{2}{n}} \nonumber \\
        &= \frac{n}{2}N_{\mbb{K}}(\tilde{x})^{\frac{2}{n}} \geq \frac{n}{2}N(\mfk{I})^{\frac{2}{n}},
    \end{align}
    where (a) is due to the fact that $\sigma_i=\bar{\sigma}_{\frac{n}{2}+i}$ and (b) follows from the AM-GM inequality. Therefore, the minimum distance of this constellation can be lower-bounded by
    \begin{equation}\label{eqn:d_s_lower_bound_comp}
        d_{\mc{S}} \geq \sqrt{\frac{n}{2}}N(\mfk{I})^{1/n}.
    \end{equation}
    Moreover, the above bound also indicates that $d_0\geq\sqrt{\frac{n}{2}}N(\Ok)^{1/n}=\sqrt{\frac{n}{2}}$. Combining \eqref{eqn:R_s}, \eqref{eqn:d_s_bound}, \eqref{eqn:d_s_lower_bound_comp}, and the fact that $d_0=\sqrt{\frac{n}{2}}$ ($1\in\Ok$) results in
    \begin{align}\label{eqn:bound_gamma}
        \frac{20\log_{10}N(\mfk{I})}{\log_2 N(\mfk{I})} &\leq \Gamma(\mc{C},\mc{S}) \nonumber \\
        &\hspace{-3pt}\leq \frac{20\log_{10}N(\mfk{I}) + 10\log_{10}|\Delta_{\mbb{K}}|\left(\frac{2}{\pi}\right)^n}{\log_2 N(\mfk{I})} ,
    \end{align}
    and thus
    \begin{equation}
        6\leq \Gamma(\mc{C},\mc{S})\leq 6 + \tilde{\gamma}_{\mc{S}}\quad \text{dB/bit/dim}.
    \end{equation}

    Let us now consider using this scheme over the Rayleigh fading network. Again, we have that every $\tilde{x}\in\Ok$ and $\tilde{x}'_{\mc{S}}$ belongs to $\Pi_{k\in\mc{S}}\mfk{p}_k$ an ideal of $\Ok$ for any $\mc{S}\subset\{1,\ldots,K\}$; therefore, the diversity order of $\mc{C}$ and $\mc{C}$ with $w_{\mc{S}}$ fixed are the same. We then have $D(\mc{C})=n/2$ and $D(\mc{C},\mc{S})=n/2$ for a totally complex $\Ok$ \cite{joseph96}. On the other hand, for a general $\Ok$ totally complex, there is not much we can say about the minimum product distance.
\end{proof}
One example of this class is shown in Fig.~\ref{fig:Qn5_fig} where we plot SNR versus symbol error rate for using the proposed scheme in Example~\ref{exp:noPID} over the AWGN network. One observes that at symbol error rate of $10^{-5}$, revealing either $w_1$ or $w_2$ results in a roughly 12 dB side information gain, which corresponds to $12/\frac{1}{2}\log_2(7)\approx 8.55$ dB/bit/dim. For the considered parameters, the above theorem says that $6 \leq \Gamma(\mc{C},\mc{S}) \leq 9.237$ dB/bit/dim.

\begin{figure}
    \centering
    \includegraphics[width=3.5in]{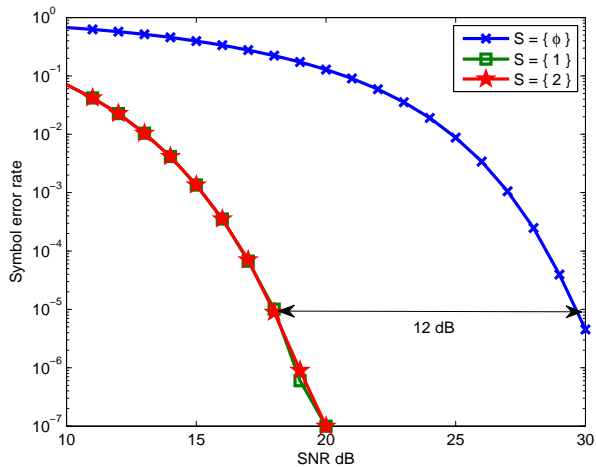}
    \caption{SNR versus symbol error rate over the AWGN network. The scheme is described in Example~\ref{exp:noPID}.}
    \label{fig:Qn5_fig}
\end{figure}

From Theorem~\ref{thm:si_gain_real} and Theorem~\ref{thm:si_gain_complex}, one observes that for the lattice index code constructed over either a totally real number field or a totally complex number field, the gap between upper and lower bounds vanishes as $N(\mfk{p}_k)$ tend to infinity. We therefore have the following corollary.
\begin{corollary}
    For the proposed lattice index code constructed over a totally real number field and that constructed over a totally complex number field, the side information is \textit{asymptotically} uniform in $N(\mfk{p}_k)$. i.e., for any $\mc{S}\in\{1,\ldots,K\}$, $\Gamma(\mc{C})\rightarrow 6$ dB as $N(\mfk{p}_k)\rightarrow \infty,~\forall\{1,\ldots,K\}$.
\end{corollary}

For imaginary quadratic integers that happen to be PID, we show that an exactly uniform side information gain of 6 dB/bit/dim can be attained in the following theorem.
\begin{theorem}\label{thm:si_gain_Q}
    Let $\mbb{K}=\mbb{Q}(\sqrt{d})$ an imaginary quadratic field with $d<0$ square-free integer whose ring of integers $\Ok$ happens to be a PID. i.e., $d\in\{-1,-2,-3,-7,-11,-19,-43,-67,-163\}$. For every $\mc{S}\subset \{1,\ldots,K\}$, we have $\Gamma(\mc{C},\mc{S})=6$ dB/bit/dim.
\end{theorem}
Before we prove this theorem, we note that imaginary quadratic fields are totally complex with signature $(0,1)$. For quadratic integers, $\Delta_{\mbb{K}}=4d$ if $d\equiv 2,3\hspace{-3pt}\mod 4$ and $\Delta_{\mbb{K}}=d$ if $d\equiv 1\hspace{-3pt}\mod 4$. Hence, the smallest gap between the bounds in Theorem~\ref{thm:si_gain_complex} happens when $d=-3$ and is roughly $0.884/\log_2 N(\mfk{I})$. Theorem~\ref{thm:si_gain_Q} shows that for some of these rings, one can further close the gap and achieve 6 dB/bit/dim exactly.
\begin{proof}
    Since $\Ok$ is a PID, every ideal is generated by a singleton in $\Ok$; specifically, $\mfk{I} = \phi\Ok$ for some $\phi\in\Ok$. Moreover, an imaginary quadratic field is a totally complex number field with degree 2; hence, $(r_1,r_2)=(0,1)$ and $\Psi(.)=\sigma_1(.)$. Let $\tilde{x}_1, \tilde{x}_2\in \Ok$ whose squared distance is $\|\Psi(\tilde{x}_1)-\Psi(\tilde{x}_2)\|^2$. One has that $\Psi(\phi \tilde{x}_1), \Psi(\phi \tilde{x}_2)\in \Psi(\phi\Ok)$ with squared distance
    \begin{align}
        d^2 &= \| \Psi(\phi \tilde{x}_1) - \Psi(\phi \tilde{x}_2) \|^2 \nonumber \\
        &= \| \sigma_1(\phi \tilde{x}_1) - \sigma_1(\phi \tilde{x}_2) \|^2 \nonumber \\
        &\overset{(a)}{=} \| \sigma_1(\phi)\sigma_1(\tilde{x}_1) - \sigma_1(\phi)\sigma_1(\tilde{x}_2) \|^2 \nonumber \\
        &\overset{(b)}{=} \|\sigma_1(\phi)\|^2 \|\sigma_1(\tilde{x}_1)-\sigma_1(\tilde{x}_2)\|^2 \nonumber \\
        &= N(\phi) \|\Psi(\tilde{x}_1)-\Psi(\tilde{x}_2)\|^2,
    \end{align}
    where (a) follows from the fact that $\sigma_1$ is a homomorphism and (b) is because Euclidean norm is multiplicative.

    Now, pick $\tilde{x}_1,\tilde{x}_2\in\Ok$ having the minimum distance. i.e., $\|\Psi(\tilde{x}_1)-\Psi(\tilde{x}_2)\|=d_0$. The squared distance between $\phi \tilde{x}_1$ and $\phi \tilde{x}_2$ will be $d^2=N(\phi) d_0^2$. Note that this is the minimum one in $\phi \Ok$ because any other pair $\phi \tilde{x}_1'$ and $\phi \tilde{x}_2'$ has
    \begin{align}
        \| \Psi(\phi \tilde{x}_1') - \Psi(\phi \tilde{x}_2') \|^2 &= N(\phi)\|\Psi(\tilde{x}_1')-\Psi(\tilde{x}_2')\|^2 \nonumber \\
        &\geq N(\phi) d_0^2.
    \end{align}

    Therefore,
    \begin{equation}
        \Gamma(\mc{C},\mc{S}) = \frac{20\log_{10}N(\phi)}{\log_2 N(\phi)}=6\quad \text{dB/bit/dim}.
    \end{equation}
\end{proof}
An example of this kind can be found in Fig.~\ref{fig:Qn7_fig} where we plot SNR versus symbol error rate for using the scheme in Example~\ref{exp:PID} over the AWGN network. One observes that at symbol error rate of $10^{-5}$, revealing messages $w_1$ and $w_2$ provides SNR gains of roughly 8.5 dB and $10.5$ dB, respectively. Hence, the side information gains of revealing $w_1$ and $w_2$ are roughly $8.5/\frac{1}{2}\log_2(7)\approx 6.055$ and $10.5/\frac{1}{2}\log_2(11)\approx 6.07$ dB/bit/dim, respectively.

\begin{figure}
    \centering
    \includegraphics[width=3.5in]{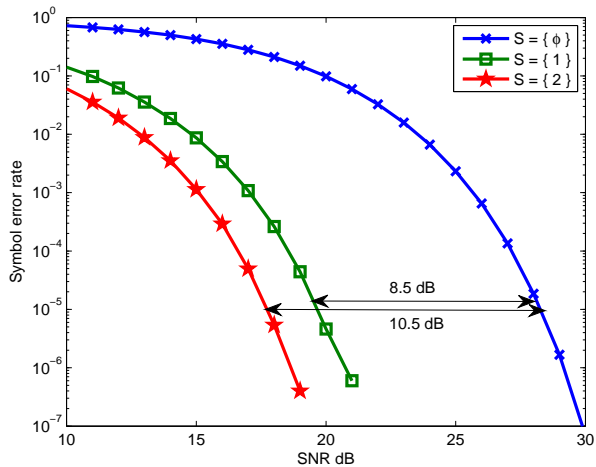}
    \caption{SNR versus symbol error rate over the AWGN network. The scheme is described in Example~\ref{exp:PID}.}
    \label{fig:Qn7_fig}
\end{figure}

\begin{remark}
    Since $\Zi$ and $\Zw$ are imaginary quadratic integers with $d=-1$ and $d=-3$, respectively, the lattice coding scheme in \cite{viterbo14index} is subsumed as a special case of the proposed scheme and Theorem~\ref{thm:si_gain_Q} partially recovers the results in \cite{viterbo14index}. The proposed scheme is in fact a generalization of their scheme to algebraic number fields and substantially expands the design space. On the other hand, there is another construction in \cite{viterbo14index} using Hurwitz integers (a ring of integers of Hamilton's quaternions) which form a non-commutative PID. The construction proposed in this paper only considers commutative rings $\Ok$ and hence does not contain this construction.
\end{remark}

\section{Design Examples}\label{sec:design}
In this section, we study some interesting designs of the proposed lattice index coding scheme. Our examples are motivated by an observation made in \cite{viterbo14index} that the scheme in \cite{viterbo14index} will have messages $w_k$ from different fields. However, in many applications, one would like the messages to be from the same fields.
To relax this, we first observe that the proposed lattice index code only hinges on the CRT which requires the corresponding prime ideals to be relatively prime. However, unlike in $\mbb{Z}$, in a general ring of algebraic integers, two prime ideals are relatively prime does not necessarily mean that they are lying above different prime numbers. Specifically, let $p$ be a prime, designs in this section consider the prime decomposition of the principal ideal $p\Ok$ where all the prime factors are relatively prime. One can then have messages all from $\mbb{F}_p$ and use the ring isomorphism to map symbols to the constellation $\Ok/p\Ok$.

To allow a fixed number of messages $K$ from the same field $\mbb{F}_p$ in our scheme, what we need is a number field $\mbb{K}$ in which $p$ splits into at least $K$ prime ideals.
Note that for a $\mbb{K}$ having degree $n$, a nature prime $p$ splits into at most $n$ prime ideals. From now on, we particularly look at $K=n$ and look for a $p$ that splits completely in $\mbb{K}$.

The first example focuses on the quadratic fields. The second one considers the cyclotomic extension in which certain primes split completely for a fixed $n$. The third design example considers the maximal totally real subfield of a cyclotomic field where again one can easily determine completely splitting primes. For the background knowledge of cyclotomic extensions and their maximal totally real subfields, the reader is referred to standard textbooks such as \cite{washington97}. For the last one, we borrow the existence result in \cite{guruswami03} which shows that there exists a number field having an infinite Hilbert class field tower in which some primes split completely all the way up the tower.

It is worth noting that very recently, Natarajan \textit{et al.} proposed in \cite{viterbo15_index_QAM} another class of codes for the same problem studied here. The codes therein use quadrature amplitude modulation (QAM) and off-the-shelf LDPC codes to achieve large side information gains for $K\leq 5$. This result is of practical interest as each message now can be from a (possibly the same) Galois field with the size a power of 2. However, these index codes are obtained via exhaustive search and hence may not be easy to generalize to any $K$. Moreover, no bounds on side information are shown in \cite{viterbo15_index_QAM}.

\subsection{Quadratic Fields}
A quadratic field is an algebraic number field $\mbb{K}$ of degree $n=2$ over $\mbb{Q}$. Particularly, one may write $\mbb{K}=\mbb{Q}(\sqrt{d})$ where $d\in\mbb{Z}$ is square free. We say $\mbb{K}$ is an imaginary quadratic field if $d<0$ and a real quadratic feild if $d>0$. Let $\mbb{K}=\mbb{Q}(\sqrt{d})$, one has its ring of integers $\mfk{O}_{\mbb{K}}=\mbb{Z}[\xi]$ given by
\begin{equation}
    \xi = \left\{\begin{array}{ll}
    \sqrt{d},                                           & d\equiv 2,3\mod 4, \\
    \frac{1+\sqrt{d}}{2},                                    & d\equiv 1\mod 4.\\
    \end{array} \right.
\end{equation}
Examples~\ref{exp:real_PID},~\ref{exp:noPID}, and \ref{exp:PID} are instances of such $\Ok$. Also, $\Delta_\mbb{K}=4d$ if $d\equiv 2,3\mod 4$ and $\Delta_\mbb{K}=d$ if $d\equiv 1\mod 4$. In general, Theorem~\ref{thm:si_gain_real} or Theorem~\ref{thm:si_gain_complex} can be applied depends on the sign of $d$. It is well-known that there are 9 such $\Ok$ are PIDs (corresponding to $d\in\{-1,-2,-3,-7,-11,-19,-43,-67,-163\}$) for which Theorem~\ref{thm:si_gain_Q} guarantees an uniform side information gain. Moreover, observe that when $d=-1$ we have the Gaussian integers and when $d=-3$ we have the Eisenstein integers. Hence, this design example partially subsumes the results in \cite{viterbo14index} as special cases.

\begin{remark}\label{rmk:best_quad}
    Note that among the rings of quadratic integers, $\Zw$ will in general provide the best performance when used over the AWGN network as it is the best packing in $\mbb{R}^2$ \cite{conway1999sphere}. Hence it seems that there's no need to pursue other rings in this class. However, having these rings in the repository can still be very useful because 1) prime numbers behave differently in different rings, 2) when used over the Rayleigh fading network, a ring of real quadratic integers would provide a diversity gain of 2, and 3) in some scenarios, with limited feedback, one may choose rings other than $\Zw$ according to the feedback to achieve better performance (one of such a possibility has been discussed in the context of function computation \cite{huang15_ACF}).
\end{remark}

The following two Lemma (whose proofs can be found in standard textbook of algebraic number theory) allow us to efficiently categorize the behavior of primes and the corresponding prime ideals in an quadratic field.

\begin{lemma}\label{lma:prime_category}
    Let $p$ be a rational prime. For a quadratic field $\mbb{K}$, one has
    \begin{itemize}
        \item if $\left(\frac{\Delta_{\mbb{K}}}{p}\right)=0$, then $p$ ramifies in $\Ok$,
        \item if $\left(\frac{\Delta_{\mbb{K}}}{p}\right)=1$, then $p$ splits in $\Ok$,
        \item if $\left(\frac{\Delta_{\mbb{K}}}{p}\right)=-1$, then $p$ remains inert in $\Ok$,
    \end{itemize}
    where $\left(\frac{\Delta_{\mbb{K}}}{p}\right)$ is the Kronecker symbol $\hspace{-3pt}\mod p$. Moreover, the Kronecker symbol $\hspace{-3pt}\mod p$ operation can be efficiently computed. (See for example \cite[Algorithm 1.4.10]{Cohen93}.)
\end{lemma}

\begin{lemma}\label{lma:prime_ideal}
    Let $p$ be an odd rational prime. For a quadratic field $\mbb{K}=\mbb{Q}(\sqrt{d})$, one has
    \begin{itemize}
        \item if $p$ ramifies in $\Ok$, then $\mfk{p}=(p,\sqrt{d})$ is a prime ideal lying above $p$,
        \item if $p$ splits in $\Ok$, then $\mfk{p}=(p,a+\sqrt{d})$ is a prime ideal lying above $p$ for any $a$ such that $a^2\equiv d\mod p$.
    \end{itemize}
    Moreover, such $a$ can be efficiently found (See for example \cite[Algorithm 1.5.1]{Cohen93}.)
\end{lemma}

One can then use a ring of quadratic integers $\Ok$ together with a splitting prime $p$ with $p\Ok=\mfk{p}\bar{\mfk{p}}$ to construct a proposed lattice index code. This will allow two users having their messages from the same field $\mbb{F}_p$. The corresponding constellations are $\Psi(\Ok/\mfk{p})$ and $\Psi(\Ok/\bar{\mfk{p}})$ for the users 1 and 2, respectively. One example can be found in Example~\ref{exp:noPID} where two users can both use coding over $\mbb{F}_7$.

\subsection{Cyclotomic Fields}
Let $\zeta_m$ be a primitive $m$th root of unity and $n=\varphi(m)$ where $\phi$ is the Euler phi function. Then $\mbb{K}_m=\mbb{Q}(\zeta_m)$ is a totally complex number field with degree $n$. Thus, Theorem~\ref{thm:si_gain_complex} applies. The ring of integers of is $\mfk{O}_{\mbb{K}_m}=\mbb{Z}[\zeta_m]$ given by
\begin{equation}
    \mbb{Z}[\zeta_m] = \{a_0+a_1\zeta_m+\ldots a_{n-1}\zeta_m^{n-1}: a_i\in\mbb{Z}\}.
\end{equation}
The discriminant of $\mfk{O}_{\mbb{K}_m}=$ is given by
\begin{equation}
    \Delta_{\mbb{K}_m} = (-1)^{\varphi(m)/2} \frac{m^{\varphi(m)}}{\underset{{p|m}}{\Pi} p^{\varphi(m)/(p-1)}}.
\end{equation}
Note that there are $n=\varphi(m)$ integers less than or equal to $m$ that is relatively prime to $m$. Call such integers $n_i$. The $n$ $\mbb{Q}$-monomorphisms are given by
\begin{equation}
    \sigma_i(\zeta_m) = \zeta_m^{n_i}.
\end{equation}
The study of cyclotomic extensions has a rich history and plays an important role in the long pursuit of the Fermat's last theorem.

Let $p$ be a natural prime: i) $p\mfk{O}_{\mbb{K}_m}$ ramifies if and only if $p|m$ and ii) if gcd$(p,m)=1$ and $f$ is the least natural number such that $p^f \equiv 1 \hspace{-3pt}\mod m$, then $p\mfk{O}_{\mbb{K}_m}=\mfk{p}_1\cdot\ldots\cdot \mfk{p}_h$ where $h\cdot f = n$ and $f$ is the inertial degree for $\mfk{p}_1,\ldots,\mfk{p}_h$. In particular, $p\Ok$ splits completely into $p\Ok = \mfk{p}_1\cdot\ldots\cdot \mfk{p}_n$ with $N(\mfk{p}_i)=p$ for $i\in\{1,\ldots,n\}$ if and only if $p\equiv 1 \hspace{-3pt}\mod m$ (the cyclotomic reciprocity law \cite[Theorem 2.13]{washington97}). Unlike quadratic integers, for a general ring of integers, prime ideals may not be easily determined as that in Lemma~\ref{lma:prime_category}. However, one can always use the roots of the minimal polynomial mod $p$ to find prime ideals (see for example \cite[Proposition 2.14]{washington97} which works for any Dedekind domain).

We can now design lattice index codes over cyclotomic integers. Consider broadcasting $K$ independent messages as described above. We first construct cyclotomic extension $\mbb{K}_m$ with degree $\phi(m)=K$. By Dirichlet's prime theorem, there are infinitely many primes $p\equiv 1 \hspace{-3pt}\mod m$ for every $m\in \mbb{Z}$. Thus, for such primes, $p\Ok$ splits completely into $K$ prime ideals. We then use those prime ideals to construct a lattice index code as described in Section~\ref{sec:proposed_LIC}. Note that for this design, all the messages would be from the same field $\mbb{F}_p$.

\begin{example}
    Let $K=4$. We choose $m=5$ and thus $\phi(5)=4$. Note that $11\equiv 1 \hspace{-3pt}\mod 5$ and hence $11\Ok$ splits completely into $4$ prime ideals $\mfk{p}_1,\ldots,\mfk{p}_4$ where $N(\mfk{p}_1)=N(\mfk{p}_2)=N(\mfk{p}_3)=N(\mfk{p}_4)=11$. One can then construct a lattice index code with these prime ideals as proposed in Section~\ref{sec:proposed_LIC}. This will allow 4 users having messages from the same field $\mbb{F}_{11}$.
\end{example}

\subsection{Maximal Totally Real Subfields of Cyclotomic Fields}
So far, we have provided two examples that are totally complex. In the following, we consider a totally real example. Let $\mbb{Q}(\zeta_m)$ be the $m$th cyclotomic field as above. $\mbb{K}_m^+ \defeq \mbb{Q}(\zeta_m+\zeta_m^{-1})$ is its maximal totally real subfield and the degree $[\mbb{Q}(\zeta_m):\mbb{K}^+] = 2$. Thus, the degree $n=[\mbb{K}_m^+:\mbb{Q}]=\varphi(m)/2$ and Theorem~\ref{thm:si_gain_real} applies. The ring of integers of $\mbb{K}_m^+$ is $\mfk{O}_{\mbb{K}_m^+} = \mbb{Z}[\zeta_m+\zeta_m^{-1}]$. Moreover, $p\mfk{O}_{\mbb{K}_m^+}$ splits completely into $\mfk{p}_1\cdot\ldots\cdot \mfk{p}_n$ with $N(\mfk{p}_i)=p$ for $i\in\{1,\ldots,n\}$ if and only if $p\equiv \pm 1 \hspace{-3pt}\mod m$. If $m$ is a natural prime, the discriminant of $\mbb{K}_m^{+}$ can be easily computed
\begin{equation}
    \Delta_{\mbb{K}_m^{+}} = m^{\frac{m-3}{2}}.
\end{equation}

\begin{example}
    Let $K=3$. We choose $m=7$ and thus $\phi(7)/2=3$. Note that $13\equiv -1 \hspace{-3pt}\mod 7$ and hence $13\mfk{O}_{\mbb{K}_m^+}$ splits completely into $3$ prime ideals $\mfk{p}_1,\ldots,\mfk{p}_3$ where $N(\mfk{p}_1)=N(\mfk{p}_2)=N(\mfk{p}_3)=13$. We can then use these prime ideals to construct a proposed lattice index code with all three messages over $\mbb{F}_{13}$.
\end{example}

\subsection{Hilbert Class Field Tower}
Notice that in the above design, we first fix $K$ and construct a lattice index code from a number field chosen according to $K$. The messages are from $\mbb{F}_p$ whose size $p$ heavily depends on $K$ and can be very large for large $K$. Here, we provide a design example where the field size $p$ can be fixed for arbitrary $K$. Before proceeding, we must note that this design is inspired by \cite{guruswami03}.

Let us start by introducing the foundation of the class field theory. The following theorem was conjectured by Hilbert in 1898 and proved by Furtw\"{a}ngler in 1930.
\begin{theorem}[Hilbert 1898 and Furtw\"{a}ngler 1930]
    For any number field $\mbb{K}$, there exists a unique finite extension $\mbb{K}'$ (called the Hilbert class field) such that i) $\mbb{K}'/\mbb{K}$ is Galois and the Galois group is isomorphic to the ideal class group of $\mbb{K}$; ii) $\mbb{K}'/\mbb{K}$ is the maximal unramified Abelian extension; iii) for any prime $\mfk{p}$, the inertial degree is the order of $\mfk{p}$ in the ideal class group of $\mbb{K}$; and iv) every ideal of $\mbb{K}$ is principal in $\mbb{K}'$.
\end{theorem}
For a prime $p$, the Hilbert $p$-class field of $\mbb{K}$ is the maximal $p$-extension (i.e., its degree is a power of $p$) $\mbb{K}_p'$ of $\mbb{K}$ contained in $\mbb{K}'$. One can construct a sequence of $p$-extensions $\{\mbb{K}_i\}$ with $\mbb{K}_i = (\mbb{K}_{i-1})_p'$ and we refer to this sequence of fields as the $p$-class field tower of $\mbb{K}_0$. The tower terminates at $i$ if it is the smallest index such that $\mbb{K}_{i+1}=\mbb{K}_i$. One can also specify a set of primes $T$ in which every prime splits completely in every field in the sequence. We call such sequence of field extensions the $T$-decomposing $p$-class field tower.

From a result in \cite[Proposition 19]{guruswami03} \cite{lenstra86}, one can construct an infinite $T$-decomposing $2$-class field tower if some mild conditions hold. In what follows, we provide an example which is borrowed from \cite[Lemma 20]{guruswami03}.

\begin{example}
    Let $d=3\cdot 5\cdot 7\cdot 11\cdot 13\cdot 17\cdot 19$ and consider the imaginary quadratic extension $\mbb{K}_0=\mbb{Q}(\sqrt{-d})$. One can show that $29 \mfk{O}_{\mbb{K}_0} = \mfk{p}_1\mfk{p}_2$ with $N(\mfk{p}_1)=N(\mfk{p}_2)=29$. Let $T=\{\mfk{p}_1,\mfk{p}_2\}$. One can show that $\mbb{K}_0$ has an infinite $T$-decomposing 2-class field tower $\mbb{K}_0\subset \mbb{K}_1 \subset\ldots$ which never terminates and $29$ splits completely all the way up the tower. For any $K$, one can then pick a totally complex field $\mbb{K}_i$ in this sequence which has a degree $n\geq K$ and $K$ prime ideals with norm $29$ to construct a lattice index code. 
\end{example}

\section{Coded Index Modulation with Constellations from Number Fields}\label{sec:CIM}
In \cite{viterbo_isit15}, Natarajan \textit{et al.} proposed the coded index modulation where powerful outer codes are concatenated with inner lattice index code (called index modulation in this scenario). The objective of the outer codes is to provide coding gains on top of side information gains provided by the index modulation. In this section, we explore this possibility and extend coded index modulation in \cite{viterbo_isit15} to algebraic number fields.

We use the method proposed in Section~\ref{sec:proposed_LIC} to obtain index modulation from a number field, and then use it in conjunction with powerful linear codes to obtain coding gains on top of side information gains. More specifically, as shown in Fig.~\ref{fig:index_CM}, one first picks an $\Ok$ and $K$ of its prime ideals lying above $p_1, \ldots, p_K$ with inertial degrees $f_1, \ldots, f_K$, respectively. The messages $w_1, \ldots, w_K$ are then from $\mbb{F}_{p_1^{f_1}}, \ldots, \mbb{F}_{p_K^{f_K}}$, respectively. The encoder then uses $N$-dimensional linear codes $C_1,\ldots, C_K$ over $\mbb{F}_{p_1^{f_1}},\ldots, \mbb{F}_{p_K^{f_K}}$, respectively, for outer coding. Each entry of the codewords $\mathbf{c}_1,\ldots,\mathbf{c}_K$ are then mapped to form $\tilde{x}$ via the mapping $\mc{M}$ in \eqref{eqn:proposed_LIC} and the embedded version $\mathbf{x}=\Psi(\tilde{x})$ is sent. Note that for each entry of $\mathbf{c}_1,\ldots,\mathbf{c}_K$, the corresponding $\mathbf{x}$ is $n$-dimensional. Thus, the transmitted signal lies in $\mbb{R}^{nN}$.
\begin{figure}
    \centering
    \includegraphics[width=2.in]{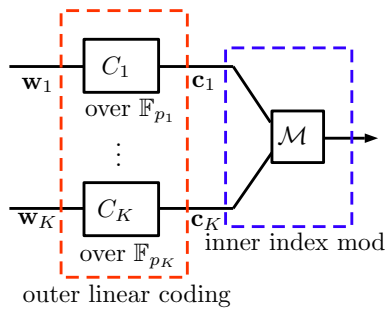}
    \caption{The index coded modulation.}
    \label{fig:index_CM}
\end{figure}

A simulation result is presented in Fig.~\ref{fig:coded} where we implement non-binary LDPC codes in conjunction with the proposed lattice index coding. Specifically, the index modulation we use in this example is the one described in Example~\ref{exp:noPID} where the constellations are isomorphic to $\mbb{F}_7\times\mbb{F}_7$. For outer coding of the two messages, we adopt two identical regular $(3,6)$-LDPC codes \cite{urbanke_book} over $\mbb{F}_7$ with progressive edge growth algorithm \cite{hu05} for selecting the parity check matrix. We set $N=4800$ and hence each message is of length $2400$ and the overall dimension of the transmitted signal is $nN=9600$. When $S=\{\phi\}$, the maximum number of iterations is set to 40 for each code and 5 iterations between two codes. When $S=\{1\}$ or $S=\{2\}$, the maximum number of iterations is set to 200. The simulation stops when 10000 symbol errors are observed. In Fig.~\ref{fig:coded}, we observe that the index coded modulation is capable of providing side information gains while enjoying sharp waterfall region resulting from coding gains offered by the outer codes. Note that the scheme adopted in Fig.~\ref{fig:coded} is by no means the best. This is also discussed in Remark~\ref{rmk:best_quad}. The purpose of this simulation is merely to demonstrate that when used in conjunction with power outer codes, the proposed scheme is able to enjoy coding gains on top of side information gains.

\begin{figure}
    \centering
    \includegraphics[width=3.5in]{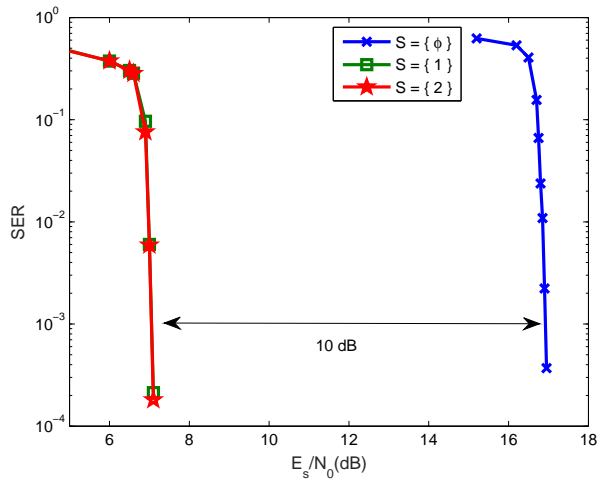}
    \caption{SNR versus symbol error rate over the AWGN network for the index coded modulation with the proposed lattice index code in Example~\ref{exp:noPID} as inner modulation and two identical (3,6) regular LDPC codes over $\mbb{F}_7$ as outer codes.}
    \label{fig:coded}
\end{figure}

\section{Concluding Remarks}\label{sec:conclude}
In this paper, the problem of broadcasting $K$ independent messages to multiple users where each of them has a subset of messages as side information has been studied. A lattice index coding scheme has been proposed which is a generalization of the scheme in \cite{viterbo14index} to general rings of algebraic integers. For some interesting classes of number fields, upper and lower bounds on the side information gains have been provided which coincide either exactly or asymptotically in message rates. This generalization has substantially expanded the design space and perhaps more importantly can provide diversity gains in addition to side information gains when used over the Rayleigh fading network. Some interesting design examples in which messages are all from the same finite field have been discussed. One potential future work is to consider a larger class of nested lattice codes in addition to self-similar ones considered in \cite{viterbo14index}. A natural extension along this line is to use Construction $\pi_A$ lattices Proposed by Huang and Narayanan in \cite{huang15piA} that has been shown able to produce good lattices. Using lattices thus constructed to achieve the capacity region of the problem studied in this paper is currently under investigation.

\section*{Acknowledgment}
The author would like to thank Prof. Krishna R. Narayanan for helpful discussions and Mr. Ping-Chung Wang for performing simulations in Fig.~\ref{fig:coded}.

\appendices
\section{Review of Algebraic Number Theory}\label{apx:prelim}
In the appendix, we provide some background knowledge on abstract algebra and algebraic number theory to facilitate discussion followed by this section. We also provide some standard properties which will be useful later on. All the properties are provided without proofs and the reader is referred to standard textbooks (for example \cite{Hungerford74} \cite{lang94}) for details.

\subsection{Algebra}
Let $\mc{R}$ be a commutative ring. An \textit{integral domain} is a commutative ring with identity and no zero divisors. An additive subgroup $\mfk{I}$ of $\mc{R}$ satisfying $ar\in\mfk{I}$ for $a\in\mfk{I}$ and $r\in\mc{R}$ is called an \textit{ideal} of $\mc{R}$. An ideal generated by a singleton is called a \textit{principal ideal}. A \textit{principal ideal domain} (PID) is an integral domain in which every ideal is principal. Let $a, b\in\mc{R}$ and $\mfk{I}$ be an ideal of $\mc{R}$; then $a$ is congruent to $b$ \textit{modulo} $\mfk{I}$ if $a-b\in\mfk{I}$. The coset decomposition $\mc{R}/\mfk{I}$ forms a ring and is called the quotient ring.

A proper ideal $\mfk{p}$ of $\mc{R}$ is said to be a \textit{prime ideal} if for $a, b\in\mc{R}$ and $ab\in\mfk{p}$, then either $a\in\mfk{p}$ or $b\in\mfk{p}$. A proper ideal $\mfk{I}$ of $\mc{R}$ is said to be a \textit{maximal ideal} if $\mfk{I}$ is not contained in any strictly larger proper ideal. Let $\mc{R}_1$ and $\mc{R}_2$ be rings. A function $\sigma:\mc{R}_1\rightarrow \mc{R}_2$ is a \textit{ring homomorphism} if
\begin{align}
    \sigma(a + b) &= \sigma(a) + \sigma(b) ~\forall a,b\in\mc{R}_1 \text{~and} \\
    \sigma(a\cdot b) &= \sigma(a)\cdot \sigma(b),~\forall a,b\in\mc{R}_1.
\end{align}
A homomorphism is said to be \textit{monomorphism} if it is injective and \textit{isomorphism} if it is bijective. Let $\mc{R}$ be a commutative ring, and $\mfk{I}_1,\ldots,\mfk{I}_n$ be ideals in $\mc{R}$. Then, from CRT, we have
\begin{equation}
    \mc{R}/\cap_{i=1}^n\mfk{I}_i \cong \mc{R}/\mfk{I}_1\times\ldots\times\mc{R}/\mfk{I}_n,
\end{equation}
where we use $\cong$ to denote ``isomorphic" and $\times$ to denote Cartesian product. Moreover, if $\mfk{I}_1,\ldots,\mfk{I}_n$ are relatively prime, then $\cap_{i=1}^n\mfk{I}_i = \Pi_{i=1}^n \mfk{I}_i$.



\subsection{Algebraic Numbers and Algebraic Integers}
An algebraic number is a root of some polynomial with coefficients in $\mbb{Z}$. A number field $\mbb{K}$ is a field extension of finite degree $[\mbb{K}:\mbb{Q}]$ and $\mbb{K}=\mbb{Q}(\theta)$ for some algebraic number $\theta$. An algebraic integer is a complex number which is a root of some monic polynomial (whose leading coefficient is 1) with coefficients in $\mbb{Z}$. The set of all algebraic integers forms a subring $\mc{B}$ of $\mbb{C}$. For any number field $\mbb{K}$, we write $\mfk{O}_{\mbb{K}}=\mbb{K}\cap\mc{B}$ and call $\mfk{O}_{\mbb{K}}$ the ring of integers of $\mbb{K}$.

For $\mbb{K}=\mbb{Q}(\theta)$ a number field of degree $n$ over $\mbb{Q}$, there are exactly $n$ distinct $\mbb{Q}$-monomorphism $\sigma_i:\mbb{K}\rightarrow\mbb{C}$. We denote by $(r_1,r_2)$ the signature of $\mbb{K}$ if among those $n=r_1 + 2r_2$ $\mbb{Q}$-monomorphisms, there are $r_1$ real $\mbb{Q}$-monomorphisms and $r_2$ pairs of complex $\mbb{Q}$-monomorphisms which are complex conjugate to each other. Moreover, for $\alpha\in\mbb{Q}(\theta)$, $\sigma_i(\alpha)$ for $i\in\{1,2,\ldots,n\}$ are the distinct zeros in $\mbb{C}$ of the minimal polynomial of $\alpha$ over $\mbb{Q}$. We call those $\sigma_i(\alpha)$ the \textit{conjugates} of $\alpha$ and define the \textit{norm} of $\alpha$ to be the product of conjugates as
\begin{equation}
    N_{\mbb{K}}(\alpha)=\prod_{i=1}^n\sigma_i(\alpha).
\end{equation}

Let $\{\alpha_1,\ldots,\alpha_n\}$ be a $\mbb{Q}$-basis for $\mbb{K}$. We define the \textit{discriminant} of $\{\alpha_1,\ldots,\alpha_n\}$ as
\begin{equation}
    \Delta[\alpha_1,\ldots,\alpha_n]\defeq \det\left(
                                                 \begin{array}{cccc}
                                                   \sigma_1(\alpha_1) & \sigma_1(\alpha_2) & \ldots & \sigma_1(\alpha_n) \\
                                                   \sigma_2(\alpha_1) & \sigma_2(\alpha_2) & \ldots & \sigma_2(\alpha_n) \\
                                                   \vdots & \vdots & \ddots & \vdots \\
                                                   \sigma_n(\alpha_1) & \sigma_n(\alpha_2) & \ldots & \sigma_n(\alpha_n) \\
                                                 \end{array}
                                               \right)^2.
\end{equation}
If $\{\alpha_1,\ldots,\alpha_n\}$ is a $\mbb{Z}$-basis for $\mfk{O}_\mbb{K}$, we define the discriminant of $\mbb{K}$ to be $\Delta_\mbb{K}\defeq \Delta[\alpha_1,\ldots,\alpha_n]$ which is invariant to the choice of $\mbb{Z}$-basis. Let $\mfk{I}$ be an ideal of $\mfk{O}_{\mbb{K}}$. The norm of $\mfk{I}$ is $N(\mfk{I})\defeq|\mfk{O}_{\mbb{K}}/\mfk{I}|$. Moreover, if $\{\beta_1,\ldots,\beta_n\}$ is a $\mbb{Z}$-basis for $\mfk{I}$, then $N(\mfk{I})=\sqrt{\frac{\Delta[\beta_1,\ldots,\beta_n]}{\Delta_{\mbb{K}}}}$. The norm is multiplicative, i.e., for two ideals $\mfk{I}_1$ and $\mfk{I}_2$ of $\Ok$, $N(\mfk{I}_1\mfk{I}_2)=N(\mfk{I}_1)N(\mfk{I}_2)$. 

Let $\mfk{p}$ be a prime ideal in $\Ok$. We say $\mfk{p}$ lies above a prime number $p$ if $\mfk{p}|p\mbb{Z}$. Since every $\Ok$ is a Dedekind domain, $p\Ok$ can be uniquely factorized into $p\Ok = \Pi_{l=1}^L \mfk{p}_l^{e_l}$ with $\mfk{p}_l$ distinct. We call $e_l$ the \textit{ramification index} of $\mfk{p}_l$ over $p$ and $f_l=[\Ok/\mfk{p}_l:\mbb{Z}/p\mbb{Z}]$ the \textit{inertial degree} of $\mfk{p}_l$ over $p$. Note that one must have $N(\mfk{p}_l)=p^{f_l}$. Also, the ramification indices and inertial degrees must satisfy $\sum_{l=1}^L e_lf_l=n$. If $e_l>1$ for some $l$, we say $p$ (or $p\Ok$ to be precise) ramifies in $\Ok$. If $L>1$, we say $p$ splits in $\Ok$. If $L=1$ and $e_1=1$ (i.e., $f_1=n$), we say $p$ remains inert in $\Ok$. One important property of the ring of integers of a number field is that every prime ideal $\mfk{p}$ is maximal and hence $\Ok/\mfk{p}\cong \mbb{F}_{p^f}$ with $f$ being the inertial degree.

\subsection{Canonical Embedding}
Here, we review the geometry induced by algebraic number fields. Consider a number field $\mbb{K}$ with degree $n$ and signature $(r_1,r_2)$. Let $\sigma_1,\ldots,\sigma_{r_1}$ be its real $\mbb{Q}$-monomorphisms and $\sigma_{r_1+1},\ldots,\sigma_n$ be the complex $\mbb{Q}$-monomorphisms where $\sigma_{r_1+r_2+i}=\bar{\sigma}_{r_1+i}$ for $i\in\{1,r_2\}$. The canonical embedding $\Psi:\mbb{K}\rightarrow \mbb{R}^{r_1}\times\mbb{C}^{r_2}\cong \mbb{R}^n$ is defined by\footnote{Note that the canonical mapping in the conference version of this paper contains typos which led to some errors. In this version, we have fixed this issue.}
\begin{equation}\label{eqn:embedding}
    \Psi(x) = (\sigma_1(x),\ldots,\sigma_{r_1}(x),\sigma_{r_1+1}(x),\ldots,\sigma_{r_1+r_2}(x)),
\end{equation}
for $x\in\mbb{K}$ and $\Psi$ is a ring homomorphism.

One can now use the canonical embedding to map $\Ok$ or ideals in $\Ok$ to lattices which we refer to as ideal lattices. Let $\{\alpha_1,\ldots,\alpha_n\}$ be a $\mbb{Z}$-basis for $\mfk{O}_\mbb{K}$, then $\Lambda_{\Ok}\defeq\Psi(\Ok)$ is a lattice in $\mbb{R}^{r_1}\times\mbb{C}^{r_2}$ with a basis $\{\Psi(\alpha_1),\ldots,\Psi(\alpha_n)\}$. Similarly, for an ideal $\mfk{I}\in\Ok$ with a $\mbb{Z}$-basis $\{\beta_1,\ldots,\beta_n\}$, $\Lambda_{\mfk{I}}\defeq \Psi(\mfk{I})$ is a lattice in $\mbb{R}^{r_1}\times\mbb{C}^{r_2}$ with a basis $\{\Psi(\beta_1),\ldots,\Psi(\beta_n)\}$.

\section{Construction using Lattices over Number Fields}\label{apx:lic_NF}
In Section~\ref{sec:proposed_LIC}, what we have proposed is merely an $n$-dimensional modulation scheme instead of a coding scheme. Here, akin to \cite{viterbo14index}, we extend the scheme proposed above over $\Ok$ to higher dimensions and construct a lattice index code from $\Ok$-lattices.

Let $\Ok$ be the ring of integers of a number field with degree $n$ and signature $(r_1,r_2)$. Let $\mfk{p}_1,\ldots,\mfk{p}_K$ be prime ideals lying above $p_1,\ldots,p_K$, respectively, and $\mc{M}$ be the ring isomorphism discussed above. Now consider a $m$-dimensional $\Ok$-lattices $\tilde{\Lambda}\defeq \mathbf{\tilde{G}}\cdot\Ok^m$ where $\mathbf{\tilde{G}}$ is the generator matrix. Also, define its sub-lattices $\tilde{\Lambda}_k\defeq \mathbf{\tilde{G}}\cdot \left(\Pi_{k'=1,k'\neq k}^K \mfk{p}_{k'}\right)^m$ and $\tilde{\Lambda}_s\defeq \mathbf{\tilde{G}}\cdot \left(\Pi_{k'=1}^K \mfk{p}_{k'}\right)^m$. The lattices considered are the above ones embedded into the Euclidean space as $\Lambda = \Psi(\tilde{\Lambda})$, $\Lambda_k = \Psi(\tilde{\Lambda}_k)$ for $k\in\{1,\ldots,K\}$, and $\Lambda_s = \Psi(\tilde{\Lambda}_s)$. We have the following decomposition
\begin{align}
    \Lambda &= \Psi(\mathbf{\tilde{G}}\cdot\Ok^m) \nonumber \\
    &\overset{(a)}{=} \mathbf{G} \cdot \Psi(\Ok^m) \nonumber \\
    &\overset{(b)}{=} \mathbf{G} \cdot \Psi\left( \mc{M}\left(\mbb{F}_{p_1^{f_1}}^m,\ldots,\mbb{F}_{p_K^{f_K}}^m\right)+(\Pi_{k'=1}^K \mfk{p}_{k'})^m \right) \label{eqn:lattice_decomp} \\
    &\overset{(c)}{=} \mathbf{G} \cdot \Psi\left( \sum_{k=1}^K (v_k \Pi_{k'=1,k'\neq k}^K \mfk{p}_{k'})^m +(\Pi_{k'=1}^K \mfk{p}_{k'})^m \right) \nonumber \\
    &= \sum_{k=1}^K \Lambda_k + \Lambda_s,
\end{align}
where in (a) $\mathbf{G}$ is an $mn\times mn$ matrix which consists of $m^2$ $n\times n$ sub-matrices $\mathbf{G}_{ij}=\diag(\Psi(\tilde{G}_{ij}))$ where $\tilde{G}_{ij}$ is the $i$th row $j$th column element of $\mathbf{\tilde{G}}$ for $1\leq i,j\leq m$, (b) follows from the decomposition of $\Ok$ in \eqref{eqn:decompose}, and in (c) $v_k$ is again the coefficient of the B\`{e}zout identity.

It is clear that $\Lambda_s\subset\Lambda_k\subset \Lambda$ for every $k\in\{1,\ldots,K\}$ and hence one can talk about coset decomposition. The proposed lattice index coding scheme uses this fact and is then given by
\begin{equation}\label{eqn:c_lattice}
    \mc{C} = \Lambda/ \Lambda_s = \sum_{k=1}^K \Lambda_k/\Lambda_s,
\end{equation}
where by $\Lambda_k/\Lambda_s$ we mean a complete set of coset representatives such that $\mc{C}$ would have the minimum energy.

We now prove some properties of lattice index codes thus constructed. We first note that for every $\mc{S}\subset\{1,\ldots,K\}$, it can be easily seen from \eqref{eqn:lattice_decomp} that
\begin{align}
    R_{\mc{S}} &= \frac{1}{mn} \log_2( \Pi_{l\in\mc{S}} p_k^{mf_k}) \nonumber \\
    &= \frac{1}{n} \sum_{k\in\mc{S}} \log_2( N(\mfk{p}_k)).
\end{align}
Moreover, note that from \eqref{eqn:c_lattice}, one has
\begin{equation}
    \mathbf{x} = \sum_{k=1}^K \lambda_k,
\end{equation}
where $\lambda_k\in\Lambda_k/\Lambda_s$. Given a $\mc{S}\in \{1,\ldots,K\}$, the signal belongs to
\begin{equation}
    \mathbf{x}_{\mc{S}}=\sum_{k\in\mc{S}}\lambda_k + \sum_{k\notin\mc{S}} \Lambda_k/\Lambda_s,
\end{equation}
which is a shifted version of
\begin{align}\label{eqn:x_k}
    &\sum_{k\notin\mc{S}} \Lambda_k/\Lambda_s \nonumber \\
    &= \mathbf{G} \cdot \Psi\left( \sum_{k\notin \mc{S}} (c_k \Pi_{k'=1,k'\neq k}^K \mfk{p}_{k'})^m \right)/ \mathbf{G} \cdot \Psi\left( (\Pi_{k'=1}^K \mfk{p}_{k'})^m \right) \nonumber \\
    &= \mathbf{G} \cdot \Psi\left( (\Pi_{k'\in\mc{S}} \mfk{p}_{k'})^m\right)/ \mathbf{G}\cdot \Psi\left( (\Pi_{k'=1}^K \mfk{p}_{k'})^m \right),
\end{align}
where the last equality follows from the CRT. In the following, we again provide a tight bound for imaginary quadratic integers which happen to be PIDs. This slightly generalized the result in \cite[Lemma 3]{viterbo14index}.

\begin{theorem}\label{thm:si_gain_Q_lattice}
    Let $\mbb{K}=\mbb{Q}(\sqrt{d})$ an imaginary quadratic field with $d<0$ square-free integer whose ring of integers $\Ok$ happens to be a PID. i.e., $d\in\{-1,-2,-3,-7,-11,-19,-43,-67,-163\}$. Let $\mc{C}$ be a lattice index code constructed using lattice over $\Ok$. For every $\mc{S}\subset \{1,\ldots,K\}$, we have $\Gamma(\mc{C},\mc{S})=6$ dB/bit/dim.
\end{theorem}
\begin{proof}
    For PIDs, \eqref{eqn:x_k} can be rewritten as
    \begin{equation}
        \mathbf{G}\cdot \Psi\left( \Pi_{k'\in\mc{S}} \phi_{k'}\Ok^m \right)/\mathbf{G}\cdot \Psi\left( \Pi_{k'=1}^K \phi_{k'}\Ok^m\right),
    \end{equation}
    and hence $\mathbf{x}_{\mc{S}}$ belongs to a shifted version of $\mathbf{G}\cdot \Psi\left( \Pi_{k'\in\mc{S}} \phi_{k'}\Ok^m \right)$. Moreover, for imaginary quadratic integers, $\Psi(.) = \sigma_1(.)$. Now, note that for every $\mathbf{\tilde{x}}\in\Ok^m$ and $\phi\in\Ok$, one has
    \begin{equation}\label{eqn:norm_PID}
        \mathbf{G}\cdot  \Psi\left( \phi\mathbf{\tilde{x}}\right) = \mathbf{G}\cdot \sigma_1(\phi)\mathbf{I}\cdot\Psi(\mathbf{\tilde{x}}) = \sigma_1(\phi) \mathbf{G}\cdot \Psi(\mathbf{\tilde{x}}).
    \end{equation}
    One can then show the result by following the steps in Theorem~\ref{thm:si_gain_Q}.


\end{proof}

\begin{remark}
    In \cite{viterbo14index}, Natarajan \textit{et al.} use another approach to bound the side information gains which involves the center density of lattices $\Lambda$ and $\Lambda_k$. A similar bound can be obtained straightforwardly for our scheme. This bound implies that for the proposed lattice index coding scheme over number fields, the gap is not necessarily limited by $\Delta_{\mbb{K}}$ as in Theorem~\ref{thm:si_gain_real} and Theorem~\ref{thm:si_gain_complex}; one may get a smaller gap by first constructing a dense lattice $\Lambda$ from $\Ok$ and then use $\Lambda$ for constructing lattice index codes as above. We suspect that for a sufficiently large dimension, one can use Construction A to construct very dense lattices over $\Ok$ which will shrink the gap. However, this is out of the scope of this paper and we leave it as a potential future work.
\end{remark}

\bibliographystyle{ieeetr}
\bibliography{journal_abbr,bib_lattice}

\end{document}